\documentclass[11pt,]{amsart}

\usepackage[utf8]{inputenc}
\usepackage[a4paper, margin=2cm]{geometry}
\usepackage{amsthm,amssymb,mathtools,amsfonts}
\usepackage[usenames]{xcolor}
\usepackage{hyperref}
\usepackage[foot]{amsaddr}

\title{\large Entanglement monogamy via multivariate trace inequalities}
\author{Mario Berta$^{1}$}
\address{$^1$ Institute for Quantum Information, RWTH Aachen University, Aachen, Germany}
\author{Marco Tomamichel$^{2,3}$}
\address{$^2$ Department of Electrical and Computer Engineering, National University of Singapore, Singapore}
\address{$^3$ Centre for Quantum Technologies, National University of Singapore, Singapore}

\DeclareMathOperator{\tr}{tr}
\DeclareMathOperator{\cone}{cone}
\newcommand{\cM}{\mathcal{M}}
\renewcommand{\P}{\mathbb{P}}
\newcommand{\nsf}[1]{\textnormal{\textsf{#1}}}

\newtheorem{lemma}{Lemma}
\newtheorem{theorem}[lemma]{Theorem}

\newtheorem{proposition}[lemma]{Proposition}

\DeclareMathOperator*{\argmin}{arg\,min}

\begin{document}

\begin{abstract}
    Entropy is a fundamental concept in quantum information theory that allows to quantify entanglement and investigate its properties, for example its monogamy over multipartite systems. Here, we derive variational formulas for relative entropies based on restricted measurements of multipartite quantum systems. By combining these with multivariate matrix trace inequalities, we recover and sometimes strengthen various existing entanglement monogamy inequalities. In particular, we give direct, matrix-analysis-based proofs for the faithfulness of squashed entanglement by relating it to the relative entropy of entanglement measured with one-way local operations and classical communication, as well as for the faithfulness of conditional entanglement of mutual information by relating it to the separably measured relative entropy of entanglement. We discuss variations of these results using the relative entropy to states with positive partial transpose, and multipartite setups. Our results simplify and generalize previous derivations in the literature that employed operational arguments about the asymptotic achievability of information-theoretic tasks.
\end{abstract}

\maketitle


\section{Introduction}

For tripartite discrete probability distributions $P_{ABC}$, the mutual information of $A$ and $B$ conditioned on $C$
can be written as the relative entropy distance to either the closest Markov chain $A-C-B$ or to the closest state that can be recovered from the marginal $P_{AC}$ by acting only on $C$. More precisely, we can rewrite the mutual information into the following variational forms (see, e.g.~\cite{linden08})
\begin{align}
I(A:B|C)_P
&= H(AC)_P+H(BC)_P-H(C)_P-H(ABC)_P \\
&= \min_{Q_{B|C}} D(P_{ABC}\|Q_{B|C} P_{AC})\label{eq:classical-recovery} \\
&= \min_{Q_{A-C-B}}D(P_{ABC}\|Q_{A-C-B}) \label{eq:classical-markov} \,,
\end{align}
where $D(P\|Q)=\sum_x P(x) (\log P(x) - \log Q(x))$ is the Kullback-Leibler divergence (or relative entropy) and $H(A)_P = -\sum_x P_A(x) \log P_A(x)$ is the Shannon entropy. Here, in the expression~\eqref{eq:classical-recovery}, the joint distribution $Q_{B|C} P_{AC}$ can be interpreted as the output of a recovery channel $Q_{B|C}$ with access to $C$ (but not $A$); the expression is minimized when $Q_{B|C} = P_{B|C}$. The minimization in the expression~\eqref{eq:classical-markov} is over all distributions with a Markov chain structure $A-C-B$; the expression is minimized when $Q_{A-C-B} = P_{B|C} P_{AC}$. As a consequence, using the non-negativity of the Kullback-Leibler divergence, one finds $I(A:B|C)_P\geq0$, which is equivalent to strong sub-additivity (SSA) of entropy.

More generally, for tripartite quantum states $\rho_{ABC}$, one defines the quantum conditional mutual information as
\begin{align}
I(A:B|C)_\rho=H(AC)_\rho+H(BC)_\rho-H(C)_\rho-H(ABC)_\rho
\end{align}
with the von Neumann entropy $H(A)_\rho=-\tr\left[\rho_A\log\rho_A\right]$. A highly non-trivial argument by Lieb and Ruskai from the seventies \cite{lieb73,lieb73a} then shows that due to entanglement monogamy the SSA inequality $I(A:B|C)_\rho\geq0$ still holds in the quantum case.

In recent years, the quantum information community has seen a lot of progress on understanding potential refinements of SSA for quantum states, with the goal of mimicking the classical version of Eqs.~\eqref{eq:classical-markov} and \eqref{eq:classical-recovery} for quantum states and quantum channels. Firstly, one can simply rewrite \cite{bertawilde14}
\begin{align}
I(A:B|C)_\rho&=\min_{\sigma_{AC},\omega_{BC}}\max_{\tau_C}D(\rho_{ABC}\|\exp(\log \sigma_{AC}+\log \omega_{BC}-\log\tau_C))\\
&=D(\rho_{ABC}\|\exp(\log\rho_{AC}+\log\rho_{BC}-\log\rho_C))\label{eq:exp-state}
\end{align}
in terms of the Umegaki's quantum relative entropy $D(\rho\|\sigma)=\tr[\rho(\log\rho-\log\sigma)]$, but due to non-commutativity any interpretation in terms of quantum Markov chains remains largely unclear \cite{Carlen14}.

Secondly, in general, we have for the alternative local recovery map form
\begin{align}
I(A:B|C)_\rho\ngeq\min_{\mathcal{R}_{C\to BC}}D(\rho_{ABC}\|(\mathcal{I}_A\otimes\mathcal{R}_{C\to BC})(\rho_{AC}))\,,\label{eq:quantum-recovery}
\end{align}
where $\mathcal{R}_{C\to BC}$ denotes quantum channels \cite{Fawzi18}. However, a series of results, first by Fawzi \& Renner \cite{fawzirenner14} and then in \cite{brandao14,berta15,wilde15,sutter16a,Junge18,sutter16,Berta21}, revealed that weaker forms of Eq.~\eqref{eq:quantum-recovery} still hold, e.g.,
\begin{align}\label{eq:cqmi-donald}
I(A:B|C)_\rho\geq\min_{\mathcal{R}_{C\to BC}} D_{\nsf{ALL}}(\rho_{ABC}\|(\mathcal{I}_A\otimes\mathcal{R}_{C\to BC})(\rho_{AC}))
\end{align}
in terms of Donald's measured relative entropy \cite{donald86},
\begin{align}
D_{\nsf{ALL}}(\rho\|\sigma)=\max_{\mathcal{M}}D(\mathcal{M}(\rho)\|\mathcal{M}(\sigma))\,,
\end{align}
with the maximum over positive operator-valued measure (POVM) measurement channels $\mathcal{M}$.\footnote{Note that Donald's original definition only considered projective measurements, but this is in fact sufficient~\cite{berta17}. Moreover, $D_{\nsf{ALL}}(\rho\|\sigma)\leq D(\rho\|\sigma)$, strictly if and only if $\rho$ and $\sigma$ do not commute \cite{berta17}.} A regularized version in terms of the quantum relative entropy distance then also follows from the asymptotic achievability of the measured relative entropy \cite{hiai91,Berta21}. Compared to the bound in Eq.~\eqref{eq:exp-state}, the bound in Eq.~\eqref{eq:cqmi-donald} lifts the classical Markov picture of approximately recovering the state with a local recovery map $P_{B|C}$ applied to the marginal $P_{AC}$ to the quantum setting via $(\mathcal{I}_A\otimes\mathcal{R}_{C\to BC})(\rho_{AC})$ (see \cite{Sutter18} and references therein).

Thirdly, a suitable generalization of an exact quantum Markov chain was established via the SSA equality condition \cite{hayden04}
\begin{align}\label{eq:QMC}
I(A:B|C)_\sigma=0\quad \iff \quad\sigma_{ABC}=\bigoplus_kp_k\sigma^k_{AC_k^L}\otimes\sigma^k_{C_k^RB}\,,
\end{align}
with respect to some induced direct sum decomposition $C=\bigoplus_kC_k^L\otimes C_k^R$. Unfortunately, lower bounding the quantum conditional mutual information in terms of the distance to exact quantum Markov chains neither works for relative entropy distance \cite{linden08}, nor for regularized relative entropy distances, nor for measured relative entropy distance~\cite{Christandl12-2}.

Now, in the context of the quantum conditional mutual information based entanglement measure squashed entanglement \cite{christandl04}, it is of importance that for an exact quantum Markov state, the reduced state $\sigma_{AB}=\sum_kp_k\sigma^k_A\otimes\sigma^k_B$ is separable\,---\,as can be easily checked using Eq.~\eqref{eq:QMC}. Then, even though the quantum relative entropy is monotone under the partial trace over $C$, still, in general
\begin{align}
I(A:B|C)_\rho\ngeq\min_{\sigma_{AB} \in \nsf{Sep}(A:B)}D_{\nsf{ALL}}(\rho_{AB}\|\sigma_{AB})
\end{align}
and the same for regularized versions thereof \cite{Christandl12-2}. Only relaxing even further and employing locally measured quantum distance measures \cite{matthews09} and in particular locally measured quantum relative entropies \cite{piani09}, one finds that \cite{brandao11,Li14,liwinter14}
\begin{align}\label{eq:cqmi-sep}
I(A:B|C)_\rho\geq\min_{\sigma_{AB} \in \nsf{Sep}(A:B)}D_{\nsf{LOCC}_1(A \to B)}(\rho_{AB}\|\sigma_{AB})\,,
\end{align}
where $\nsf{LOCC}_1(A \to B)$ denotes measurements that use a single round of communication: They first measure out $A$ and then perform a conditional measurement on the system $B$ depending on the measurement outcome on $A$. Even though such measurements have a reduced distinguishing power \cite{matthews09,Lami18,Lancien16}, crucially, they are still tomographically complete, and thus the right-hand side is zero if and only if $\rho_{AB}$ is separable.

Going back to the bigger picture, the two types of refined SSA bounds as in Eqs.~\eqref{eq:cqmi-donald} and~\eqref{eq:cqmi-sep} seem in general incompatible, but are both entanglement monogamy inequalities with widespread applications in quantum information science (see aforementioned reference and references therein). Moreover, for the former type, a unified matrix analysis based proof approach has emerged. Namely, extending Lieb and Ruskai's original argument for the proof of SSA~\cite{lieb73a,lieb73}, the first step is to employ the multivariate Golden-Thompson inequalities from \cite{sutter16,hiai17,Sutter17}: For any $n\in \mathbb{N}$, Hermitian matrices $\{H_k\}_{k=1}^n$, and any $p \geq 1$, one has
\begin{align}\label{eq:GT-multi}
\log \left\|\exp \left( \sum_{k=1}^n H_k \right) \right\|_p \leq \int_{-\infty}^{\infty} \textnormal{d}\beta_0(t)\log\left\|\prod_{k=1}^n  \exp\left( (1+i t) H_k \right)\right\|_p\,,
\end{align}
where $\|\cdot\|_p$ denotes the Schatten $p$-norm and $\beta_0(t)=\frac{\pi}{2}\left(\cosh(\pi t)+1\right)^{-1}$ is a fixed probability density on~$\mathbb{R}$. The second step is then to combine this with dual variational representations of quantum entropy in terms of matrix exponentials \cite{berta17,berta23}.

In contrast, the previously known proofs of the refined SSA bound from Eq.~\eqref{eq:cqmi-sep} are based on involved operational arguments about the asymptotic achievability of information-theoretic tasks \cite{brandao11,Li14}, including the asymptotic achievability of quantum state redistribution \cite{devetak08,devetak09}, partial state merging \cite{Yang08}, and Stein’s lemma in hypothesis testing \cite{Li14,brandao13}.\footnote{The conceptually different work \cite{liwinter14} gives extendability refinements of SSA based on iterating Markov refinements of SSA and then combining these bounds with finite quantum de Finetti theorems with quantum side information \cite{Christandl07} to make the connection with separability.}

Here, we seek after a unified matrix analysis based proof for Eq.~\eqref{eq:cqmi-sep} and other entanglement monogamy inequalities of similar type. For this, we derive novel variational formulas for quantum relative entropies based on restricted measurements, which then, indeed, enable us to employ a similar, matrix analysis approach in terms of multivariate Golden-Thompson inequalities. Namely, the core step in our derivations is to employ the multivariate Eq.~\eqref{eq:GT-multi} for $n=3,4,5,6$ and $p=1,2$. Importantly, this allows us to fully bypass the previously employed operational arguments from quantum information theory. Consequently, we give concise proofs that lead to tight SSA separability refinements and other new entanglement monogamy inequalities, including positive partial transpose bounds and multipartite extensions. On the way we further derive various strengthened recoverability bounds, such as for the conditional entanglement of mutual information and the multipartite squashed entanglement. In turn, the explicit form of our novel entanglement monogamy inequalities also feature recoverability maps, revealing a deeper connection between SSA separability refinements and SSA recoverability bounds.

The rest of the manuscript is structured as follows. In Section \ref{sec:measured}, we derive new variational formulas for locally measured quantum relative entropies. In Section \ref{sec:inequalities} we present the derivations of our entanglement monogamy inequalities around the SSA separability refinements from Eq.~\eqref{eq:cqmi-sep}. This is in terms of squashed entanglement (Section \ref{sec:squashed}), relative entropy of entanglement (Sections \ref{sec:relative} and \ref{sec:piani}), conditional entanglement of mutual information (Section \ref{sec:cemi}), as well as for multipartite extensions thereof (Section \ref{sec:multipartite}). In Section \ref{sec:outlook} we then conclude with some outlook on open questions.


\section{On measured divergences and entanglement measures}
\label{sec:measured}

We start by introducing some notational conventions used in this work. Throughout we assume that Hilbert spaces, denoted $A$, $B$, $C$, etc., are finite-dimensional and quantum states are positive semi-definite operators with unit trace acting on such spaces, or tensor product spaces of them. We use subscripts to indicate what spaces an operator acts on and by convention when we introduce an operator $X_{AB}$ acting on $A \otimes B$ we implicitly also introduce its marginals $X_A$ and $X_B$, defined via the respective partial traces of $X_{AB}$ over $B$ and $A$, respectively. We often omit identity operators, e.g., $X_A Y_{AB}$ should be understood as the matrix product $(X_A \otimes 1_B) Y_{AB}$. Functions are applied on the spectrum of an operator coinciding with the domain of the function, which means that $X_A^{-1}$ is the generalized inverse and $\log(X_A)$ is always bounded. At various points we employ indices $x$, $y$ or $z$ that are meant to be taken from discrete index sets $\mathcal{X}$, $\mathcal{Y}$ and $\mathcal{Z}$ that are understood to be defined implicitly. We use $\geq$ and $>$ to denote the L\"owner order on operators, e.g., an operator $L$ is positive semi-definite if and only if $L \geq 0$, and a positive semi-definite operator $L$ has full support if and only if $L > 0$.


\subsection{Definitions and some properties}

Consider a quantum state $\rho > 0$ and an operator $\sigma > 0$. 
We recall the definition and variational formula for Umegaki relative entropy between $\rho$ and $\sigma$ as
\begin{align}
    D(\rho\|\sigma) &:= \tr \big( \rho (\log \rho - \log \sigma) \big) \\
    &\ =\sup_{\omega > 0}\ \tr (\rho \log \omega) - \log \tr (\exp (\log \sigma + \log \omega)\,,
\end{align}
Here the optimization is over all operators $\omega$ with full support, a set that is clearly not closed. Nonetheless, the supremum is taken as $\omega = \exp(\log \rho - \log \sigma)$. We can extend the definition to general states (without full support) by taking an appropriate continuous extension, namely\footnote{The argument on the right-hand side is monotone in $\epsilon$ due to the joint convexity of the relative entropy, and the supremum thus constitutes a limit (in case of convergence).}
\begin{align}
    D(\rho\|\sigma) := \sup_{\epsilon \in (0, 1] } D\big((1-\epsilon)\rho + \epsilon \pi \big\| (1-\epsilon) \sigma + \epsilon \pi \big) \label{eq:continuous-extension} ,
\end{align}
where $\pi$ is the completely mixed state.
We note that the above quantity is finite if and only if $\rho \ll \sigma$, i.e., if the support of $\rho$ is contained in the support of $\sigma$. In the following we will always assume full support in our definitions and use Eq.~\eqref{eq:continuous-extension} to extend to the general case where needed.

Based on this we arrive at the definition of the relative entropy of entanglement for a bipartite quantum state $\rho_{AB}$ and the bipartition $A : B$, which is given by
\begin{align}
    E(A:B)_{\rho} := \min_{\sigma_{AB} \in \nsf{Sep}(A:B)} \ D(\rho_{AB} \| \sigma_{AB})\,,
\end{align}
where $\nsf{Sep}(A:B)$ denotes the set of separable states on the bipartition $A:B$, i.e.\ quantum states that decompose as $\sigma_{AB} = \sum_x Y_A^x \otimes Y_B^x$ for positive semi-definite operators $\{ Y_A^x \}_x$ and $\{ Y_B^x \}_x$. Here, the minimum is always taken since $D(\cdot\|\cdot)$ is jointly convex and continuous in $\sigma_{AB}$ as long as we stay away from the (uninteresting, as we are seeking a minimum) boundary where $\rho_{AB} \not\ll \sigma_{AB}$.

We will also use various notions of measured relative entropy. In the following $\cM$ is a set of POVMs, and its elements $M = \{ M^z \}_z$ are sets of positive semi-definite operators satisfying $\sum_z M^z = 1$.

For example, \nsf{ALL} denotes the set of all POVMs. If the states are bipartite on $A$ and $B$, we consider various specialized sets. On the one hand, the sets $\nsf{SEP}(A:B)$ and $\nsf{PPT}(A:B)$ contain POVMs whose elements are separable ($\nsf{SEP}$) or have positive partial transpose ($\nsf{PPT}$), respectively. On the other hand, elements of $\nsf{LOCC}(A:B)$ are operationally defined as a POVMs that can be implemented by local operations and finite classical communication ($\nsf{LOCC}$). Elements of $\nsf{LOCC}_1(A \to B)$ are POVMs that only use a single round of communication: they first measure out $A$ and then perform a conditional measurement on the system $B$ depending on the measurement outcome on $A$. Without loss of generality, such measurements can be written in the form
\begin{align}
    \left\{ M_{AB}^z \right\}_z \quad \textnormal{with} \quad M_{AB}^z = \sum_x Q_A^x \otimes Q_B^{z|x}\,, \label{eq:locc1meas}
\end{align}
where $Q_A^x \geq 0$ and $Q_B^{z|x} \geq 0$ with $\sum_x Q_A^x = 1_A$ and $\sum_z Q_B^{z|x} = 1_B$. Here $x$ labels the data sent from Alice to Bob whereas $z$ is the final output after Bob's measurement. Finally, the set $\nsf{LO}(A:B)$ allows only local measurements without communication, which are of the form $M_{AB}^{z} = Q_A^x \otimes Q_B^y$, where $z = (x,y)$ collects the local outputs.

With this in hand, let us define a measured relative entropy and an entanglement measure for each $\cM$ described above:
\begin{align}
    &D_{\cM}(\rho \| \sigma) := \sup_{M \in \cM} D( \P_{\rho, M} \| \P_{\sigma, M} )\\
    &E_{\cM}(A:B)_{\rho} := \min_{\sigma_{AB} \in \nsf{Sep}(A:B)} D_{\cM}(\rho_{AB} \| \sigma_{AB})\,.
    \label{eq:meas}
\end{align}
Here, $\P_{\rho,M}(z) = \tr(\rho M_z)$ is the probability mass function emanating from Born's rule. We note that the minimum is achieved as $D_{\cM}$, a supremum of jointly convex functions, is itself jointly convex and thus, as argued above, the minimum is taken.

From the inclusions $\nsf{ALL} \supseteq \nsf{PPT}(A:B) \supseteq \nsf{SEP}(A:B) \supseteq \nsf{LOCC}(A:B) \supseteq \nsf{LOCC}_1(A \to B) \supseteq \nsf{LO}(A:B)$ we trivially get
\begin{align}
    E(A:B)_{\rho} &\geq E_{\nsf{ALL}}(A:B)_{\rho} \geq E_{\nsf{PPT}}(A:B)_{\rho} \geq E_{\nsf{SEP}}(A:B)_{\rho} \notag\\
    &\geq E_{\nsf{LOCC}}(A:B)_{\rho} \geq E_{\nsf{LOCC}_1}(A \to B)_{\rho} \geq E_{\nsf{LO}}(A:B)_{\rho} \,,
\end{align}
with the shorthand $E_{\nsf{LOCC}_1}(A \to B)_{\rho}:=E_{\nsf{LOCC}_1(A\to B)}(A:B)_{\rho}$. We further introduce PPT variants defined as
\begin{align}
    &P_{\cM}(A:B)_{\rho} := \min_{\sigma_{AB} \in \nsf{ppt}(A:B)} D_{\cM}(\rho_{AB} \| \sigma_{AB})\,,
\end{align}
where $\nsf{ppt}(A:B)$ denotes the set of states that have positive partial transpose withe respect to the bipartition $A:B$, which we study in particular in combination with measurements $\cM=\nsf{PPT}$. We note that all of the above quantities are faithful since $\nsf{LO}(A:B)$ is already tomographically complete. Further, there are minimax statements available that interchange the supremum over the set of measurements with the infimum over the set of states \cite[Lemma 13]{brandao13}.

Above quantities are in general not additive on tensor product states and one can then write down the regularization
\begin{align}
    E_{\cM}^{\infty}(A:B)_{\rho} :=\lim_{n\to\infty} \frac{1}{n}E_{\cM}(A:B)_{\rho^{\otimes n}}\,,
\end{align}
which are well-defined, with operational interpretations in terms of optimal asymptotic quantum Stein's error exponents for the corresponding restricted class of measurements \cite[Theorem 16]{brandao13}. In general, it is unclear how to make quantitative statements about the regularization, but for the class $\nsf{ALL}$ we have the following~\cite[Lemma 2.4]{Berta21}.

\begin{lemma}\label{lem:asymptotic-achiev}
For any n-partite quantum state $\rho_{A^n}$ and $\sigma_{A^n} \geq 0$ with $\rho_{A^n} \ll \sigma_{A^n}$, we have
\begin{align}
    D\left(\rho_{A^n}\middle\|\sigma_{A^n} \right)-\log|\mathrm{spec}(\sigma_{A^n})|\leq D_{\nsf{ALL}}\left(\rho_{A^n}\middle\|\sigma_{A^n}\right)\leq D\left(\rho_{A^n}\middle\|\sigma_{A^n}\right)\,,
\end{align}
where $|\mathrm{spec}(\sigma_{A^n})|\leq\mathrm{poly}(n)$ when $\sigma_{A^n}$ is invariant under permutations of the $n$ systems.
\end{lemma}
This is an extension of the asymptotic achievability of the measured relative entropy \cite{hiai91} and follows from the pinching inequality \cite{hayashi02b} together with Schur-Weyl duality showing that the number of distinct eigenvalues of $\sigma_{A^n}$ only grows polynomial in $n$ (see, e.g., \cite[Lemma 4.4]{harrowphd}).

Finally, one can also define multipartite extensions of above quantities. For example, we have the tripartite separable measured relative entropy of entanglement $E_{\nsf{SEP}}(A:B:C)_{\rho}$ and its regularization, $E_{\nsf{SEP}}^{\infty}(A:B:C)_{\rho}$. We will not directly use multipartite versions of $\nsf{LOCC}_1(A \to B)$ and hence we do not discuss its different variations \cite{brandao13,Li15}.

We should verify that all these entanglement measures are indeed entanglement monotones, i.e., monotone under application of $\nsf{LOCC}(A:B)$ completely positive and trace preserving (cptp) maps. It is easy to see, and well-known, that $E_{\cM}$ with $\cM \in \big\{ \nsf{ALL}, \nsf{SEP}(A:B), \nsf{PPT}(A:B), \nsf{LOCC}(A:B)\big\}$ are entanglement monotones. This is no longer true for $\cM = \nsf{LOCC}_1(A \to B)$. Instead, we show the following, weaker, statement.\footnote{One suspects that this is well-known, but we could not find a reference.}

\begin{lemma}\label{lem:lodpi}
    Both $D_{\nsf{LOCC}_1(A \to B)}(\cdot \| \cdot )$ and $E_{\nsf{LOCC}_1}(A \to B)$ are monotone under $\nsf{LOCC}_1(A\to B)$ operations, i.e.\ under local operations supported by one-way communications from $A$ to $B$.
\end{lemma}

\begin{proof}
    Without loss of generality a measurement in $\nsf{LOCC}_1(A \to B)$ is of the form Eq.~\eqref{eq:locc1meas}. To show the monotonicity under an $\nsf{LOCC}_1(A \to B)$ operation we only need to show that the above structure of the measurement is preserved under the adjoint operation. Again, without loss of generality, we can write a $\nsf{LOCC}_1(A \to B)$ operation in the form $\mathcal{G} = \sum_k \mathcal{E}_k \otimes \mathcal{F}_k$ where $\mathcal{F}_k:B \to B'$ are cptp maps and $\mathcal{E}_k: A \to A'$ are completely positive trace non-increasing (cptni) maps forming an instrument, such that $\sum_k \mathcal{E}_k$ is cptp again. Given a measurement in $\nsf{LOCC}(A':B')$ can now construct a measurement on $\nsf{LOCC}(A:B)$ with the matrices
    \begin{align}
        E_{AB}^y = \sum_{k,x} \underbrace{\mathcal{E}_k^{\dagger}(Q_{A'}^x)}_{=:\ \tilde{Q}_{A}^{(k,x)}} \otimes \underbrace{\mathcal{F}_k^{\dagger}(Q_{B'}^{y|x})}_{=:\ \tilde{Q}_{B}^{y|(k,x)}} 
    \end{align}
    that has the same structure as in Eq.~\eqref{eq:locc1meas}. Namely, we can verify that
    \begin{align}
        \sum_y \tilde{Q}_{B}^{y|(k,x)} = \sum_y  \mathcal{F}_k^{\dagger}(Q_{B'}^{y|x}) = \mathcal{F}_k^{\dagger}(1_{B'}) = 1_{B}, \qquad
        \sum_{k,x} \tilde{Q}_{A}^{(k,x)} = \sum_{k,x} \mathcal{E}_k^{\dagger}(Q_{A'}^x) = \sum_k \mathcal{E}_k^{\dagger}(1_{A'}) = 1_{A} \,.
    \end{align}
    Hence, $D_{\nsf{LOCC}_1(A \to B)}(\rho_{AB} \| \sigma_{AB}) \geq D_{\nsf{LOCC}_1(A \to B)}(\mathcal{G}(\rho_{AB}) \| \mathcal{G}(\sigma_{AB}))$, and since this holds for all separable states $\sigma_{AB}$ and $\mathcal{G}$ preserves this structure, the desired result for $E_{\nsf{LOCC}_1}(A\,;B)_{\rho}$ also follows.
\end{proof}

Moreover, using similar arguments, one can verify that $D_{\nsf{LO}(A:B)}(\rho_{AB}\|\sigma_{AB})$ and $E_{\nsf{LO}}(A:B)_{\rho}$ are monotone under local operations.


\subsection{General variational formulas}

Our approach is to employ dual representations of quantum entropy as in \cite{berta23,sutter16}. For that, we explore variational expressions for measured relative entropies. For unrestricted measurements, we have the well-known expression
\begin{align}
    D_{\nsf{ALL}}(\rho\|\sigma) = \sup_{\omega > 0} \left\{ \tr [\rho \log \omega] - \log \tr [\sigma \omega] \right\} \,, \label{eq:mrelent}
\end{align}
which is in fact consistent with Eq.~\eqref{eq:continuous-extension} without assumptions on the support of $\rho$ or $\sigma$, and will be finite if and only if $\rho \ll \sigma$. For other classes of measurements we can show the following generic bound.

\begin{lemma} \label{lem:variationalbound}
    Define $C_{\cM}$ as the union of the cones spanned by the POVM elements of measurements in $\cM$, i.e.,
    $C_{\cM} := \bigcup_{M \in\cM} \cone{} \{M^z\}_z$.
    Then, for a quantum state $\rho$ and any $\sigma \geq 0$, we have
    \begin{align}         
    D_{\cM}(\rho\|\sigma) \leq 
    \sup_{\omega > 0 \atop \omega \in C_{\cM}} \left\{
    \tr [\rho \log \omega] - \log \tr [\sigma \omega] \right\} \,.
    \end{align}
\end{lemma}  

The proof is an adaptation of the argument in~\cite{berta17}.

\begin{proof}
    We first treat the case where both $\rho$ and $\sigma$ have full support, we thus have $\frac{\P_{\rho,M}(z)}{\P_{\sigma,M}(z)} \in (0, \infty)$.

    Using the operator Jensen's inequality, we can bound the measured relative entropy as follows:
    \begin{align}
        D_{\cM}(\rho\|\sigma) &= \sup_{M \in \cM} \left\{ \sum_z \P_{\rho,M}(z) \log \frac{\P_{\rho,M}(z)}{\P_{\sigma,M}(z)} \right\} \\
        &= \sup_{M \in \cM} \left\{ \tr \left[ \rho \sum_z \sqrt{M^z} \log \left(  \frac{\P_{\rho,M}(z)}{\P_{\sigma,M}(z)} \right) \sqrt{M^z} \right] \right\} \\
        &\leq \sup_{M \in \cM} \left\{ \tr \left[ \rho \log \left( \sum_z \frac{\P_{\rho,M}(z)}{\P_{\sigma,M}(z)} M^z  \right) \right] \right\} \\
        &\leq \sup_{\omega > 0 \atop \omega \in C_{\cM}} \left\{ \tr \left[\rho \log \omega \right] - \log \tr \left[ \sigma  \omega \right] \right\} \,, \label{eq:varitionalbound}
    \end{align}
    where, in order to establish the last inequality, we used the fact that
    \begin{align}
        \tr \left[ \sigma \sum_z  \frac{\P_{\rho,M}(z)}{\P_{\sigma,M}(z)}  M^z\right] = \sum_z \P_{\rho,M}(z) = 1 \,
    \end{align}
    that $\omega = \sum_z \frac{\P_{\rho,M}(z)}{\P_{\sigma,M}(z)} M^z \in \cone{} \{M^z\}_z$ by definition of the cone, and that $\omega>0$ as $\frac{\P_{\rho,M}(z)}{\P_{\sigma,M}(z)} \in (0, \infty)$ together with $M\in\cM$.

    For the general case we simply note that the right-hand side of Eq.~\eqref{eq:varitionalbound} is jointly convex in $(\rho, \sigma)$ and vanishes for $(\pi, \pi)$, and thus
    \begin{align}
        \sup_{\epsilon \in (0, 1]} D_{\cM}( (1-\epsilon) \rho + \epsilon \pi \| (1-\epsilon) \sigma + \epsilon \pi) &\leq \sup_{\epsilon \in (0, 1]} (1-\epsilon) \sup_{\omega > 0 \atop \omega \in C_{\cM}} \tr \left[ \rho \log \omega \right] - \log \tr \left[\sigma  \omega \right], 
    \end{align}
    from which the result immediately follows.
\end{proof}

We note that $C_{\nsf{ALL}}$ is the cone of positive semi-definite operators, and from Eq.~\eqref{eq:mrelent} we know that equality in the above lemma holds. For other sets of measurements we do not always have a good characterization of the respective set (which might not even be convex in general), but $C_{\nsf{SEP}(A:B)}$ and $C_{\nsf{PPT}(A:B)}$ are comprised of separable positive semi-definite operators and positive semi-definite operators with positive partial transpose, respectively. We do not know if equality in Lemma \ref{lem:variationalbound} holds for either $\nsf{SEP}(A:B)$ or $\nsf{PPT}(A:B)$.


\subsection{Cone for local measurements and constrained communication}

On first look, note that the set $C_{\nsf{LOCC}_1(A \to B)}$ is comprised of positive semi-definite operators of the form
\begin{align}
    \omega_{AB} = \sum_x Q_{A}^x \otimes \omega_{B}^x
\end{align}
where $\omega_B^x \geq 0$ and $Q_A^x \geq 0$ such that $\sum_x Q_A^x = 1_A$, and $x$ goes over some finite alphabet. However, the upper bound we get using this in Lemma~\ref{lem:variationalbound} does not appear to be tight. We can, however, show the following exact variational formula for the $\nsf{LOCC}_1(A \to B)$ measured relative entropy.

\begin{lemma} \label{lem:oneway-exact}
    Let $A'$ be isomorphic to $A \otimes A$ and consider the set of operators
    \begin{align}
        C_{A'B}^* := \left\{ \sum_{x=1}^{d^2} P_{A'}^x \otimes \omega_B^x \,\middle|\,\forall {x,x'}: P_{A'}^x \geq 0 \land  P_{A'}^x P_{A'}^{x'} = \delta_{xx'} P_{A'}^x \land \omega_B^x \geq 0  \right\}\,.
    \end{align}
    (These are operators that are classical-quantum in some basis on $A'$.)
    Then, with $\rho_{A'B}$ and $\sigma_{A'B}$ consistent embeddings of $\rho_{AB} > 0$ and $\sigma_{AB} > 0$, respectively, we have
    \begin{align}
        D_{\nsf{LOCC}_1(A \to B)}(\rho_{AB} \| \sigma_{AB}) 
        = \sup_{ \omega_{A'B} > 0 \atop \omega_{A'B} \in C_{A'B}^*} \Big\{ \tr \left[ \rho_{A'B} \log \omega_{A'B} \right] - \log \tr \left[ \sigma_{A'B} \omega_{A'B} \right] \Big\} \,.
    \end{align}
    Moreover, the optimal measurement is comprised of a (rank-$1$) POVM on $A$ with at most $d^2$ outcomes, followed by a conditional projective measurement on $B$.
\end{lemma}

\begin{proof}
    We first note that due the joint convexity of $D(\cdot\|\cdot)$ we know that the optimal measurement on $A$ is extremal. From~\cite[Theorem 2.21]{holevo12} follows that extremal POVMs have at most $d^2$ rank-1 elements, where $d$ is the dimension of $A$. In particular, via Naimark's dilation, there exists a rank-$1$ projective measurement on $A'$ that produces the same statistics.
    
    Since $D_{\nsf{LOCC}_1(A \to B)}(\rho_{AB} \| \sigma_{AB}) = D_{\nsf{LOCC}_1(A' \to B)}(\rho_{A'B} \| \sigma_{A'B})$ due to the data-processing inequality for local operations in Lemma~\ref{lem:lodpi}, we can restrict the optimization over measurements for the latter quantity to POVMs with elements of the form
    \begin{align}
        M_{A'B}^z = \sum_{x=1}^{d^2} P_{A'}^x \otimes Q_B^{z|x} ,
    \end{align}
    where $P_{A'}^x \geq 0$ are rank-1 with $P_{A'}^x P_{A'}^{x'} = \delta_{xx'} P_{A'}^x$ as in the definition of $C_{A'B}^*$. Applying the series of steps in the proof of Lemma~\ref{lem:variationalbound} we arrive at the bound
    \begin{align}
        D_{\nsf{LOCC}_1(A \to B)}(\rho_{AB} \| \sigma_{AB}) 
        &\leq \sup_{ \omega_{A'B} \in C_{A'B}^*} \big\{ \tr \left[ \rho_{A'B} \log \omega_{A'B} \right] - \log \tr \left[ \sigma_{A'B} \omega_{A'B} \right] \big\} \\
        &= \sup_{ \omega_{A'B} \in C_{A'B}^*} \big\{ \tr \left[ \rho_{A'B} \log \omega_{A'B} \right] + 1 - \tr \left[ \sigma_{A'B} \omega_{A'B} \right] \big\} .
    \end{align}
    Using the eigenvalue decomposition $\omega_{B}^x = \sum_y \lambda_{y|x} P_{B}^{y|x}$ we can write
    \begin{align}
        &\tr \left[ \rho_{A'B} \log \omega_{A'B} \right] + 1 - \tr \left[ \sigma_{A'B} \omega_{A'B} \right] \\
        &\qquad = \sum_x \sum_y \underbrace { \tr \left[ \rho_{A'B} P_{A'}^x \otimes P_{B}^{y|x} \right] }_{ = \P_{\rho,M}(x,y)} \log \lambda_{y|x} + 1 - \lambda_{y|x} \underbrace{ \tr \left[ \sigma_{A'B} P_{A'}^x \otimes P_{B}^{y|x} \right] }_{= \P_{\sigma,M}(x,y)} ,
    \end{align}
    where $M$ defines a projective measurement on $A'$ using $P_{A'}^x$ followed by a conditional projective measurement on $B$ using $P_{B}^{y|x}$.
    Optimizing the above expression over $\lambda_{y|x}$ yields
    \begin{align}
        \sum_x \sum_y \P_{\rho,M}(x,y) \log \frac{\P_{\rho,M}(x,y)}{\P_{\sigma,M}(x,y)} = D( \P_{\rho,M} \| \P_{\sigma,M} )
    \end{align}
    and, thus, we can conclude that
    \begin{align}
        \sup_{ \omega_{A'B} \in C_{A'B}^*} \big\{ \tr \left[ \rho_{A'B} \log \omega_{A'B} \right] + 1 - \tr \left[ \sigma_{A'B} \omega_{A'B} \right] \big\}
        &\leq \sup_{M} D( \P_{\rho,M} \| \P_{\sigma,M} ) \\
        &\leq D_{\nsf{LOCC}_1(A \to B)}(\rho_{AB} \| \sigma_{AB})
    \end{align}
    where the form of the measurement $M$ in the supremum can be restricted as prescribed in the statement of the lemma.
\end{proof}

Next, we discuss the case of $\nsf{LO}$ measurements. For this, let $A'$ be isomorphic to $A \otimes A$ and $B'$ be isomorphic to $B \otimes B$ and consider operators that are classical in some basis on $A'$ and $B'$, i.e., the set
    \begin{align}
        C_{A'B'}^* := \left\{ \sum_{x,y=1}^{d^2} \alpha_{x,y} P_{A'}^x \otimes P_{B'}^y \,\middle|\,\forall {x,x'}: P_{A'}^x, P_{B'}^x, \alpha_{x,x'} \geq 0 \land  P_{A'}^x P_{A'}^{x'} = \delta_{xx'} P_{A'}^x \land P_{B'}^x P_{B'}^{x'} = \delta_{xx'} P_{B'}^x  \right\}
    \end{align}
Then, with $\rho_{A'B'}$ and $\sigma_{A'B'}$ consistent embeddings of $\rho_{AB}$ and $\sigma_{AB}$ as above, we have
\begin{align}
    D_{\nsf{LO}}(\rho_{AB} \| \sigma_{AB}) 
    = \sup_{ \omega_{A'B'} > 0 \atop \omega_{A'B'} \in C_{A'B'}^*} \big\{ \tr \left[ \rho_{A'B'} \log \omega_{A'B'} \right] - \log \tr \left[ \sigma_{A'B'} \omega_{A'B'} \right] \big\} \,.
\end{align}
This characterization, as for the case of one-way communication in Lemma~\ref{lem:oneway-exact}, essentially comes from the fact that the optimal local measurements can be assumed to be (rank-$1$) POVMs with at most $d^2$ outcomes on $A'$ and $B'$.

Variations of the above arguments are also possible for more complex multi-partite measurement structures, but we leave this as an exercise for the reader who has applications of those in mind.


\subsection{Comparison with restricted Schatten one-norms}
\label{sec:norm-considerations}

Restricted Schatten one-norms leading to metrics
\begin{align}
\|\rho-\sigma\|_{\cM}:=\sup_{M \in \cM} \left\|\P_{\rho, M}-\P_{\sigma, M}\right\|_1
\end{align}
have been considered in the literature \cite{matthews09,Lami18}. Similar versions can be defined for the fidelity as well, which we denote by
\begin{align}
F_{\cM}(\rho,\sigma):=\inf_{M \in \cM} F\left(\P_{\rho, M},\P_{\sigma, M}\right)\quad\text{with}\quad F(P,Q):=\sum_x\sqrt{P(x)Q(x)}\,.
\end{align}
A couple of properties are noteworthy:
\begin{itemize}
    \item Two-outcome POVMs are optimal for $\|\rho_{AB}-\sigma_{AB}\|_{\cM}$.
    \item We have the Pinsker type inequalities
    \begin{align}
        D_{\cM}(\rho_{AB}\|\sigma_{AB})\geq-\log F_{\cM}(\rho_{AB},\sigma_{AB})\geq\frac{1}{4}\|\rho_{AB}-\sigma_{AB}\|^2_{\cM}\,.
    \end{align}
    \item Dimension dependent (and basically tight) norm equivalences to the non-restricted Schatten one-norm are available \cite{Lami22}
    \begin{align}
        \|\cdot\|_1\geq\|\cdot\|_{\nsf{PPT}}\geq\|\cdot\|_{\nsf{SEP}}\geq\|\cdot\|_{\nsf{LOCC}}\geq\|\cdot\|_{\nsf{LOCC}_1}\geq\|\cdot\|_{\nsf{LO}}\geq\frac{\|\cdot\|_1}{2\sqrt{2}d}\quad\text{for $d=\min\{|A|,|B|\}$.}
    \end{align}
     Furthermore, one has the alternative bounds \cite{Lami18}
     \begin{align}
        \text{$\|\cdot\|_{\nsf{LOCC}}\geq\frac{\|\cdot\|_1}{2d-1}$ and when $|A|\leq|B|$ also $\|\cdot\|_{\nsf{LOCC}_1(A\to B)}\geq\frac{\|\cdot\|_1}{2d-1}$ .}
     \end{align}
    \item Multipartite extensions are understood as well \cite{Lancien13}.
\end{itemize}


\section{Entropic entanglement inequalities}
\label{sec:inequalities}

\subsection{Squashed entanglement}
\label{sec:squashed}

Based on the conditional quantum mutual information (CQMI)
\begin{align}
    I(A:B|C)_{\rho}:=H(AC)_\rho+H(BC)_\rho-H(ABC)_\rho-H(C)_\rho
\end{align}
one defines {\it squashed entanglement} as \cite{christandl04}
\begin{align}
    I_{\nsf{SQ}}(A:B)_\rho:=\frac{1}{2}\inf_{\rho_{ABC}}I(A:B|C)_{\rho}\,,
\end{align}
where the infimum is over all tripartite quantum state extensions $\rho_{ABC}$ of $\rho_{AB}$ on any system $C$ (with no bound on the dimension of $C$). The following theorem implies that squashed entanglement is non-zero on entangled states \cite{brandao11,Li14,liwinter14}.\footnote{The first proof of faithfulness in \cite{brandao11} is currently incomplete in the version available online, as part of the derivations are based on an imported result from \cite{brandao10} that has a flaw \cite{berta22,Berta23-2}. However, the proof of faithfulness is fixable \cite{brandaochristandlyard} (see also \cite{berta22}).}

\begin{theorem}\label{thm:sq-main}
    Let $\rho_{ABC} > 0$ be any tripartite state. We have
    \begin{align}\label{eq:sq-mainI}
        I(A:B|C)_{\rho}\geq E_{\nsf{LOCC}_1}(B \to A)_{\rho}-\Big\{E(A:C)_{\rho}-E_{\nsf{ALL}}(A:C)_{\rho}\Big\}\,,
    \end{align}
    and consequently,
    \begin{align}\label{eq:sq-mainII}
        I_{\nsf{SQ}}(A:B)_\rho\geq\frac{1}{2}E^{\infty}_{\nsf{LOCC}_1}(B\to A)_{\rho}\geq\frac{1}{2}E_{\nsf{LOCC}_1}(B\to A)_{\rho}\,.
    \end{align}
    Moreover, the same lower bounds hold for $A\leftrightarrow B$ as $I(A:B|C)_{\rho}$ is symmetric under this exchange.
\end{theorem}

Note that strong sub-additivity (SSA) of quantum entropy corresponds to $I(A:B|C)_{\rho}\geq0$ and hence Theorem \ref{thm:sq-main} corresponds to a strengthening of SSA. The stronger single-copy version in Eq.~\eqref{eq:sq-mainI} is new. The consequence in Eq.~\eqref{eq:sq-mainII} corresponds to \cite[Theorem 2]{Li14}, which is itself a strengthening of \cite[Theorem 1]{brandao11} (see also \cite{brandao13}). One advantage of our formulation in Theorem $\ref{thm:sq-main}$ is that we have some information on the structure of the optimizer in the lower bound $E^{\infty}_{\nsf{LOCC}_1}(B\,;A)_{\rho}$, as in fact (see the proof of Theorem \ref{thm:sq-main})
\begin{align}
    I(A:B|C)_\rho \geq
    &\;\frac{1}{n}D_{\nsf{LOCC}_1}\left(\rho_{AB}^{\otimes n}\middle\|\int_{-\infty}^{\infty} \textnormal{d}\beta_0(t)\tr_{C^n}\left[\left(\rho_{BC}^{\frac{1+it}{2}}\rho_{C}^{\frac{-1+it}{2}}\right)^{\otimes n}\sigma_{A^n:C^n}\left(\rho_{C}^{\frac{-1-it}{2}}\rho_{BC}^{\frac{1-it}{2}}\right)^{\otimes n}\right]\right)\nonumber\\
    & -\frac{1}{n}\log|\mathrm{poly}(n)|
\end{align}
for any separable state optimizer $\sigma_{A^n:C^n}\in\argmin_{\sigma_{A^nC^n}\in \nsf{Sep}(A^n:C^n)} \ D(\rho_{AC}^{\otimes n} \| \sigma_{A^nC^n})$ and the probability density $\beta_0(t)=\frac{\pi}{2}\left(\cosh(\pi t)+1\right)^{-1}$. This features a recovery map and thus points to further connections between entanglement monogamy and recovery refinements of SSA. However, unfortunately this structure does not seem to further translate to the single-copy lower bound $E_{\nsf{LOCC}_1}(B\,;A)_{\rho}$. We refer to the discussion around \cite[Lemma 3.11]{sutter16a} and related results on composite hypothesis testing \cite{Berta21}.

If wanted, further standard estimates can be made on the single-copy lower bound from Eq.~\eqref{eq:sq-mainII} as done in \cite[Corollary 3.13]{lamiphd}
\begin{align}
    E_{\nsf{SQ}}(A:B)_\rho\geq\frac{1}{2}E_{\nsf{LOCC}_1}(B\to A)_{\rho}
        & \geq -\frac{1}{2}\log\sup_{\sigma_{AB}\in\nsf{Sep}(A:B)}F_{\nsf{LOCC}_1(B \to A)}(\rho_{AB},\sigma_{AB})\\
        & \geq\frac{1}{8}\inf_{\sigma_{AB}\in\nsf{Sep}(A:B)}\left\|\rho_{AB}-\sigma_{AB}\right\|_{\nsf{LOCC}_1(B \to A)}^2\\
        & \geq\frac{1}{2(4d-2)^2}\inf_{\sigma_{AB}\in\nsf{Sep}(A:B)}\left\|\rho_{AB}-\sigma_{AB}\right\|_1^2\quad\text{for $d=\min\{|A|,|B|\}$,}\label{eq:locc-one-norm}
\end{align}
following the considerations from Section \ref{sec:norm-considerations}. We state these bounds from \cite[Corollary 3.13]{lamiphd} here, as since the original proofs of similar statements \cite{brandao13,Li14,liwinter14}, the dimension dependent factors in above chain of inequalities have been improved to their optimal value as stated above \cite{Lami18}. As such, our work also supersedes the bounds from \cite[Corollary 1]{liwinter14}. Finally, as discussed in \cite{Christandl12-2}, the dimension dependent factor in Eq.~\eqref{eq:locc-one-norm} is necessary due to the anti-symmetric state example.

\begin{proof}[Proof of Theorem \ref{thm:sq-main}]
    Let us fix some slack parameter $\nu > 0$. We first prove the bound in Eq.~\eqref{eq:sq-mainI} up to this slack. We start by constructing some states and operators that we will be using in the proof. First, let us introduce 
    \begin{align}
        \sigma_{A:C} \in \argmin_{\sigma_{AC} \in \nsf{Sep}(A:C)} \ D(\rho_{AC} \| \sigma_{AC}), \qquad 
        \sigma_{A:C} = \sum_k\sigma^k_{A}\otimes\sigma^k_{C} \,,
    \end{align}
    which is a minimizer for the entanglement entropy and is separable on the partition $A:C$, as indicated in the second equality. We now introduce the space $B'$ isomorphic to $B \otimes B$ and an (arbitrary) embedding $\rho_{AB'C}$ of $\rho_{ABC}$ into this larger space. Next we apply a rotated Petz recovery map to the state $\sigma_{A:C}$ and introduce the recovered states
    \begin{align}
        \gamma^{t}_{A:B'C} := \sum_k\sigma_{A}^k\otimes\left(\rho_{B'C}^{\frac{1+it}{2}}\rho_{C}^{\frac{-1+it}{2}}\sigma^k_{C}\rho_{C}^{\frac{-1-it}{2}}\rho_{B'C}^{\frac{1-it}{2}}\right), \qquad \gamma_{A:B'C}:=\int_{-\infty}^{\infty} \textnormal{d}\beta_0(t)\gamma_{A:B'C}^{t} ,
    \end{align}
    for $t \in \mathbb{R}$. One notes that these states are separable in the bipartition $A:B'C$ by construction. We now use Lemma~\ref{lem:oneway-exact} as well as the definition of the supremum to write
    \begin{align}
    E_{\nsf{LOCC}_1}(B \to A) \leq D_{\nsf{LOCC}_1(B \to A)}(\rho_{AB'} \| \gamma_{A:B'}) 
        \leq \tr \left( \rho_{AB'} \log X_{AB'} \right) - \log \tr \left( \gamma_{A:B'} X_{AB'} \right) + \nu \,.
    \end{align}
    where $X_{AB'} \in C_{AB'}^*$ is some operator with full support that is classical on $B'$, i.e.\ it has the form
    \begin{align}
        X_{AB'} = \sum_x F_{A}^x \otimes P_{B'}^x \label{eq:projlocc} ,
    \end{align}
    where $\{ P_{B'}^x \}_x$ are orthonormal rank-$1$ projectors decomposing the identity on $B'$ and $F_{A}^x \geq 0$ are arbitrary positive semi-definite matrices. Finally, we construct the state
    \begin{align}
        \hat{\gamma}_{A:C} := \frac{\tr_{B'}\Big[\int_{-\infty}^{\infty} \textnormal{d}\beta_0(t)X_{AB'}^{\frac{1+it}{2}}\gamma^{t}_{A:B'C}X_{AB'}^{\frac{1-it}{2}}\Big]}{\tr\left[X_{AB'}\gamma_{A:B'}\right]}\,, \label{eq:gammaAC}
    \end{align}
    which inherits separability in the partition $A:C$ since
    \begin{align}
        & \tr_{B'}\left[ X_{AB'}^{\frac{1+it}{2}}\gamma^{t}_{A:B'C}X_{AB'}^{\frac{1-it}{2}}\right]\nonumber\\
            &\qquad = \sum_{x,x',k}(F_A^x)^{\frac{1+it}{2}}\sigma^k_{A}(F_A^{x'})^{\frac{1-it}{2}}\otimes\tr_{B'}\left[P_{B'}^{x}\rho_{B'C}^{\frac{1+it}{2}}\rho_{C}^{\frac{-1+it}{2}}\sigma^k_{C}\rho_{C}^{\frac{-1-it}{2}}\rho_{B'C}^{\frac{1-it}{2}}P_{B'}^{x'}\right]\\
            &\qquad =\ \sum_{x,k}(F_A^x)^{\frac{1+it}{2}}\sigma^k_{A}(F_A^{x})^{\frac{1-it}{2}}\otimes\tr_{B'}\left[P_{B'}^{x}\rho_{B'C}^{\frac{1+it}{2}}\rho_{C}^{\frac{-1+it}{2}}\sigma^k_{C}\rho_{C}^{\frac{-1-it}{2}}\rho_{B'C}^{\frac{1-it}{2}}\right]\,,
    \end{align}
    where we used that $P_{B'}^x P_{B'}^{x'} = \delta_{xx'}P_{B'}^x$ and cyclicity under $\tr_B$ to simplify the expression. In essence, the structure of $\nsf{LOCC}_1$ measurements and the respective operator $X_{AB'} \in C_{AB'}^*$ as in Eq.~\eqref{eq:projlocc} is needed here to ensure that separability is preserved and no entanglement is created in this multiplication. Finally, we introduce an operator $Y_{AC} > 0$ satisfying
    \begin{align}
        E_{\nsf{ALL}}(A:C)_{\rho} \leq D_{\nsf{ALL}}(\rho_{AC} \| \hat{\gamma}_{A:C} )
        \leq \tr \left[ \rho_{AC} \log Y_{AC} \right] - \log \tr \left[ Y_{AC} \hat{\gamma}_{A:C} \right] + \nu ,
    \end{align}
    where we used the variational formula for measured relative entropy.

    Now we have everything in place, and the proof proceeds straightforwardly. First, we write
    \begin{align}
        &I(A:B|C)_{\rho}+E(A:C)_{\rho} \nonumber\\
            &\qquad =D\left(\rho_{AB'C}\middle\|\exp(\log\rho_{B'C}+\log\sigma_{A:C}-\log\rho_C)\right)\\
            &\qquad 
            = \sup_{\omega_{AB'C} > 0}\Big\{\tr\left[\rho_{AB'C}\log\omega_{AB'C}\right]-\log\tr\left[\exp(\log\omega_{AB'C}+\log\rho_{B'C}+\log\sigma_{A:C}-\log\rho_C)\right] \Big\} \,,
    \end{align}
    where in the last step we employed the variational formula for the relative entropy. At this point we simply choose $\omega = \exp (\log X_{AB'} + \log Y_{AC})$ using the two operators defined above. This, and the five matrix Golden-Thompson inequality for the Schatten two-norm from \cite[Corollary 3.3]{sutter16} allow us to further bound
    \begin{align}
        &I(A:B|C)_{\rho}+E(A:C)_{\rho} \nonumber\\
            &\qquad \geq \tr\left[\rho_{AB'C}\log\exp(\log X_{AB'}+\log Y_{AC})\right]\nonumber\\
            &\qquad\qquad -\log\tr\left[\exp(\log X_{AB'}+\log Y_{AC}+\log\rho_{BC}+\log\sigma_{A:C}-\log\rho_C)\right] \label{eq:insertomega} \\
            &\qquad \geq \tr\left[\rho_{AB'}\log X_{AB'}\right]+\tr\left[\rho_{AC}\log Y_{AC}\right]\nonumber\\
            &\qquad \qquad -\log\tr\Big[Y_{AC}\int_{-\infty}^{\infty} \textnormal{d}\beta_0(t)X_{AB'}^{\frac{1+it}{2}}\underbrace{\sum_k\sigma_{A}^k\otimes\left(\rho_{BC}^{\frac{1+it}{2}}\rho_{C}^{\frac{-1+it}{2}}\sigma^k_{C}\rho_{C}^{\frac{-1-it}{2}}\rho_{B'C}^{\frac{1-it}{2}}\right)}_{=\ \gamma^{t}_{A:B'C}}X_{AB'}^{\frac{1-it}{2}}\Big] \\
            &\qquad = \tr\left[\rho_{AB'}\log X_{AB'}\right]+\tr\left[\rho_{AC}\log Y_{AC}\right]-\log \tr\left[X_{AB'}\gamma_{A:B'}\right]\cdot\tr\left[Y_{AC}\hat{\gamma}_{A:C}\right] \\
            &\qquad \geq E_{\nsf{LOCC}_1}(B \to A)_{\rho}+E_{\nsf{ALL}}(A:C)_{\rho} - 2\nu \,,
    \end{align}
    where the equality simply follows by substitution of~\eqref{eq:gammaAC} and the ultimate inequality follows from the definition of $X_{AB'}$ and $Y_{AC}$.
    This concludes the proof of Eq.~\eqref{eq:sq-mainI} once we leverage the fact that $\nu > 0$ can be chosen arbitrarily small.

    Next, the first step in Eq.~\eqref{eq:sq-mainII} follows from the additivity of the CQMI together with the asymptotic achievability of the measured relative entropy in Lemma \ref{lem:asymptotic-achiev}, realizing that for tensor product inputs the optimization over separable states in the definition of $E_{\cM}$ can be restricted to permutation invariant states (due to the unitary invariance and joint convexity of the relative entropy).

    Finally, the second step in Eq.~\eqref{eq:sq-mainII} can be deduced from the super-additivity~\cite[Theorem 1]{piani09}
    \begin{align}
        E_{\nsf{LOCC}_1}(B_1B_2 \to A_1A_2)_{\rho}\geq E_{\nsf{LOCC}_1}(B_1 \to A_1)_{\rho}+E_{\nsf{LOCC}_1}(B_2 \to A_2)_{\rho}\,,
    \end{align}
    noting that\,---\,in the notation of \cite{piani09}\,---\,the set of measurements $\nsf{LOCC}_1(B \to A)$ is compatible with the set of states $\nsf{Sep}(A:B)$.
\end{proof}


\subsection{Relative entropy of entanglement}
\label{sec:relative}

Previously known lower bound proofs on the CQMI proceeded via two steps of multipartite monogamy inequalities, going through the relative entropy of entanglement \cite{brandao11,Li14}. As the intermediate steps are of independent interest, we now give simple and direct proofs for strengthened single-copy versions of these bounds.

\begin{proposition}\label{prop:state-redistribution}
    Let $\rho_{ABC}$ be any tripartite state. We have
    \begin{align}\label{eq:state-redistributionI}
        I(A:B|C)_{\rho}\geq E_{\nsf{ALL}}(A:BC)_{\rho}-E(A:C)_{\rho}\,,
    \end{align}
    and consequently
    \begin{align}\label{eq:state-redistributionII}
        I(A:B|C)_{\rho}\geq E^{\infty}(A:BC)_{\rho}-E^{\infty}(A:C)_{\rho}\,,
    \end{align}
    where the regularized relative entropy of entanglement terms on the right-hand side are defined as $E^{\infty}(A:B)_{\rho} :=\lim_{n\to\infty} \frac{1}{n}E(A:B)_{\rho^{\otimes n}}$. Moreover, the same lower bounds hold for $A\leftrightarrow B$ as $I(A:B|C)_{\rho}$ is symmetric under this exchange.
\end{proposition}

We note that the stronger single-copy version in Eq.~\eqref{eq:state-redistributionI} is novel. The consequence in Eq.~\eqref{eq:state-redistributionII} is \cite[Lemma 1]{brandao11}, which was based on the asymptotic achievability of quantum state redistribution \cite{devetak08,devetak09} together with the asymptotic continuity \cite{Donald99,Radtke06} and non-lockability \cite{Horodecki05-2} of the relative entropy of entanglement. We emphasize that Eq.~\eqref{eq:state-redistributionII} was also invoked in the later proof in \cite{Li14}. In contrast, our proof is elementary via multivariate matrix trace inequalities.

\begin{proof}[Proof of Proposition \ref{prop:state-redistribution}]
    For the proof of the first bound, we use similar, but simpler arguments as in the proof of the first bound in Theorem \ref{thm:sq-main}. Namely, we employ the three matrix Golden-Thompson inequality for the Schatten two-norm in the form of \cite[Eq.~39]{sutter16}\footnote{This is equivalent to Lieb's triple matrix inequality \cite{lieb73a}, as shown in \cite[Lemma 3.4]{sutter16}.}. With a separable state optimizer
    \begin{align}
        \sigma_{A:C} \in \argmin_{\sigma_{AC}\in \nsf{Sep}(A:C)} \ D(\rho_{AC} \| \sigma_{AC}), \qquad 
        \sigma_{A:C} = \sum_k\sigma^k_{A}\otimes\sigma^k_{C} \,,
    \end{align}
    we find
    \begin{align}
        I(A:B|C)_{\rho}+E(A:C)_{\rho}& =D\left(\rho_{ABC}\middle\|\exp(\log\rho_{BC}+\log\sigma_{A:C}-\log\rho_C)\right)\\
            & \geq D_{\nsf{ALL}}\left(\rho_{ABC}\middle\|\int_{-\infty}^{\infty} \textnormal{d}\beta_0(t)\rho_{BC}^{\frac{1+it}{2}}\rho_{C}^{\frac{-1+it}{2}}\sigma_{A:C}\rho_{C}^{\frac{-1-it}{2}}\rho_{BC}^{\frac{1-it}{2}}\right)\\
            & = D_{\nsf{ALL}}\left(\rho_{ABC}\middle\|\sum_k\sigma^k_{A}\otimes\left(\int_{-\infty}^{\infty} \textnormal{d}\beta_0(t)\rho_{BC}^{\frac{1+it}{2}}\rho_{C}^{\frac{-1+it}{2}}\sigma^k_{C}\rho_{C}^{\frac{-1-it}{2}}\rho_{BC}^{\frac{1-it}{2}}\right)\right)\\
            & \geq E_{\nsf{ALL}}(A:BC)_{\rho}\,,
    \end{align}
    with the probability density $\beta_0(t)=\frac{\pi}{2}\left(\cosh(\pi t)+1\right)^{-1}$.

    The second bound follows from additivity of the quantum mutual information on tensor product states together with the asymptotic achievability of the measured relative entropy from Lemma \ref{lem:asymptotic-achiev}, in the same way as we derived the second bound in Theorem \ref{thm:sq-main}.
\end{proof}

The next relative entropy of entanglement bound is as follows.

\begin{proposition}\label{thm:post_li-winter}
    Let $\rho_{ABC}$ be any tripartite state. We have
    \begin{align}\label{eq:post_li-winter}
        E(A:BC)_{\rho} \geq E_{\nsf{LOCC}_1}(B \to A)_{\rho} + E_{\nsf{ALL}}(A:C)_{\rho}\,,
    \end{align}
    and, consequently,
    \begin{align}\label{eq:li-winter}
        &E(A:BC)_{\rho} \geq E_{\nsf{LOCC}_1}(B\to A)_{\rho} + E^{\infty}(A:C)_{\rho}\quad\text{and}\\
        &E^{\infty}(A:BC)_{\rho} \geq E^{\infty}_{\nsf{LOCC}_1}(B\to A)_{\rho} + E^{\infty}(A:C)_{\rho}\,.
    \end{align}
\end{proposition}

We note that the stronger single-copy version in Eq.~\eqref{eq:post_li-winter} is novel. We were not able to directly replace the $E_{\nsf{ALL}}(A:C)_{\rho}$ term in the lower bound with the larger $E(A:C)_{\rho}$. The first consequence in Eq.~\eqref{eq:li-winter} is \cite[Theorem 1]{Li14}, whereas the second consequence can now be combined with the regularized Eq.~\eqref{eq:li-winter} leading to
\begin{align}\label{eq:bcy}
E_{\nsf{SQ}}(A:B)_\rho\geq\frac{1}{2}E^{\infty}_{\nsf{LOCC}_1}(B \to A)_{\rho}\geq\frac{1}{2}E_{\nsf{LOCC}_1}(B\,;A)_{\rho}\,,
\end{align}
as proven directly in Theorem \ref{thm:sq-main}. 

\begin{proof}[Proof of Proposition \ref{thm:post_li-winter}]
    We first prove Eq.~\eqref{eq:post_li-winter}, which is almost analogous to the proof of Theorem \ref{thm:sq-main}, up to some simplifications. We use the same embedding of $\rho_{ABC}$ to $\rho_{AB'C}$. Let
    \begin{align}
        \sigma_{A:B'C} \in \argmin_{\sigma_{AB'C}\in \nsf{Sep}(A:B'C)} \ D(\rho_{AB'C} \| \sigma_{AB'C}), \qquad 
        \sigma_{A:B'C} = \sum_k\sigma^k_{A}\otimes\sigma^k_{B'C} \,,
    \end{align}
    be a separable state optimizer. We may express the relative entropy of entanglement using the variational formula for relative entropy
    \begin{align}
        E(A:BC)_{\rho} &= D(\rho_{AB'C} \| \sigma_{A:B'C}) \\
        &= \sup_{\omega_{AB'C} > 0} \tr\left[ \rho_{AB'C} \log \omega_{AB'C} \right] - \log \tr\left[ \exp(\log \sigma_{A:B'C} + \log \omega_{AB'C} ) \right] \,,
    \end{align}
    where $\omega_{AB'C}$ is an arbitrary positive definite matrix. We will now choose it to be of the form
    \begin{align}
        \omega_{AB'C} = \exp\left( \log X_{AB'} + \log Y_{AC}  \right)\,,
    \end{align}
    where $Y_{AC} > 0$ is general and $X_{AB'} > 0$ is of the $\nsf{LOCC}_1(A \to B)$ form in Eq.~\eqref{eq:projlocc}, both still to be optimized over. We can then bound $E(A:BC)_{\rho}$ using the three-matrix Golden-Thompson inequality as follows:
    \begin{align}
        E(A:BC)_{\rho} &\geq \sup_{X_{AB'}, Y_{AC}>0} \Big\{ \tr[ \rho_{AB'} \log X_{AB'}] + \tr[ \rho_{AC} \log Y_{AC}] \\
        &\qquad - \log \tr \left[ \exp \left( \log \sigma_{AB'C} + \log X_{AB'} + \log Y_{AC} \right) \right] \Big\} \\
        &\geq \sup_{X_{AB'}, Y_{AC}>0} \Bigg\{ \tr[ \rho_{AB'} \log X_{AB'}] + \tr[\rho_{AC} \log Y_{AC}] \notag\\
        &\qquad
        - \log  \int_{-\infty}^{\infty} \textnormal{d}\beta_0(t) \tr \left[ \sigma_{AB'C} X_{AB'}^{\frac{1 + it}{2}} Y_{AC} X_{AB'}^{\frac{1 - it}{2}} \right] \Bigg\} \,. \label{eq:thelastone}
    \end{align}
    Now, define
    \begin{align}
        \widetilde{\sigma}_{A:C} := \frac{\int_{-\infty}^{\infty} \textnormal{d}\beta_0(t) \tr_{B'} \left[ X_{AB'}^{\frac{1 + it}{2}}  \sigma_{A:B'C} X_{AB'}^{\frac{1 - it}{2}}  \right] }{ \tr \left[X_{AB'} \sigma_{A:B'} \right] }\,.
    \end{align}
    Due to the $\nsf{LOCC}_1(A \to B)$ structure of $X_{AB'}$ in Eq.~\eqref{eq:projlocc} and $\sigma_{A:B'C} = \sum_k \sigma_{A}^k \otimes \sigma_{B'C}^k$, we realize that
    \begin{align}
        \tr_B \left[ X_{AB'}^{\frac{1 + it}{2}} \sigma_{A:B'C} X_{AB'}^{\frac{1 - it}{2}} \right] 
        &= \sum_{x,x',k} \left( F_{A}^x \right)^{\frac{1 + it}{2}} \sigma_{A}^k ( F_{A}^{x'} )^{\frac{1 - it}{2}} \otimes \tr_{B'}\! \left[ P_{B'}^x \sigma_{B'C}^k P_{B'}^{x'} \right] \\
        &= \sum_{x,k} \left( F_{A}^x \right)^{\frac{1 + it}{2}} \sigma_{A}^k ( F_{A}^{x} )^{\frac{1 - it}{2}} \otimes \tr_{B'}\! \left[ P_{B'}^x \sigma_{B'C}^k \right] ,
    \end{align}
    and, thus, $\widetilde{\sigma}_{A:C}$ inherits the separable structure on the bipartition $A:C$ from $\sigma_{A:B'C}$. Using this definition we can now further bound Eq.~\eqref{eq:thelastone} to arrive at
    \begin{align}
        E(A:BC)_{\rho} &\geq \sup_{X_{AB'}, Y_{AC}} \Big\{ \tr[\rho_{AB} \log X_{AB'}] + \tr[\rho_{AC} \log Y_{AC}] 
            - \log \tr \left[ \widetilde{\sigma}_{A:C} Y_{AC} \right] \tr \left[ \sigma_{A:B'} X_{AB'} \right] \Big\} \\
        &\geq E_{\nsf{LOCC}_1}(B \to A)_{\rho} + E_{\nsf{ALL}}(A:C)_{\rho} \,.\label{eq:locc1-sep}
    \end{align}

    Finally, Eq.~\eqref{eq:li-winter} then follows by the additivity of the quantum relative entropy on product states together with the asymptotic achievability of the measured relative entropy from Lemma \ref{lem:asymptotic-achiev}.
\end{proof}


\subsection{Conditional entanglement of mutual information}
\label{sec:cemi}

The quantity
\begin{align}
I(A|\bar{A}:B|\bar{B})_\rho:=I(A\bar{A}:B\bar{B})_\rho-I(\bar{A}:\bar{B})_\rho
\end{align}
can be seen as a symmetric version of the tripartite CQMI via $I(A:B|C)_\rho\equiv I(A|C:B|C)_\rho$. It is then the basis of the entanglement measure {\it conditional entanglement of mutual information (CEMI)} \cite{Yang08}
\begin{align}
I_{\nsf{CEMI}}(A:B)_\rho:=\frac{1}{2}\inf_{\rho_{A\bar{A}B\bar{B}}}I(A|\bar{A}:B|\bar{B})_\rho\,,
\end{align}
where the infimum goes over all bipartite extensions $\rho_{A\bar{A}B\bar{B}}$ of $\rho_{AB}$ on systems $\bar{A}\bar{B}$ (with no bound on the dimensions of $\bar{A}$ and $\bar{B}$). By definition we have
\begin{align}
    I_{\nsf{CEMI}}(A:B)_\rho\geq I_{\nsf{SQ}}(A:B)_\rho
\end{align}
and CEMI shares similarly complete axiomatic entanglement measurement properties as squashed entanglement \cite{Yang08}. However, whereas no separation between $I_{\nsf{CEMI}}$ and $I_{\nsf{SQ}}$ is known, CEMI often gives more structure. For example, one finds the following recoverability lower bounds (see \cite{wilde14} for related bounds.).

\begin{proposition}\label{prop:cemi-recover}
Let $\rho_{A\bar{A}B\bar{B}}$ be any four-party state. We have
\begin{align}
    I(A|\bar{A}:B|\bar{B})_\rho & \geq D_{\nsf{ALL}}\left(\rho_{A\bar{A}B\bar{B}}\middle\|\int_{-\infty}^{\infty} \textnormal{d}\beta_0(t)\left(\mathcal{R}_{\bar{A}\to A\bar{A}}^{[t]}\otimes\mathcal{R}_{\bar{B}\to B\bar{B}}^{[t]}\right)(\rho_{\bar{A}\bar{B}})\right) \label{eq:cemi-recovI}\\
    I(A|\bar{A}:B|\bar{B})_\rho & \geq -\int_{-\infty}^{\infty} \textnormal{d}\beta_0(t)\log F\left(\rho_{A\bar{A}B\bar{B}},\left(\mathcal{R}_{\bar{A}\to A\bar{A}}^{[t]}\otimes\mathcal{R}_{\bar{B}\to B\bar{B}}^{[t]}\right)(\rho_{\bar{A}\bar{B}})\right)\,,\label{eq:cemi-recovII}
\end{align}
with local quantum channels
\begin{align}
\text{$\mathcal{R}_{\bar{A}\to A\bar{A}}^{[t]}(\cdot):=\left(\rho_{A\bar{A}}^{\frac{1+it}{2}}\rho_{\bar{A}}^{\frac{-1+it}{2}}\right)(\cdot)\left(\rho_{\bar{A}}^{\frac{-1-it}{2}}\rho_{A\bar{A}}^{\frac{1-it}{2}}\right)$ and $\mathcal{R}_{\bar{B}\to B\bar{B}}^{[t]}(\cdot)$ similar,}
\end{align}
and the probability density $\beta_0(t)=\frac{\pi}{2}\left(\cosh(\pi t)+1\right)^{-1}$.
\end{proposition}

The proof is as in \cite{sutter16,Sutter17} via multivariate trace inequalities and is given in Appendix \ref{app:missing-proofs}. Additionally, the corresponding regularized lower bound as
\begin{align}
    \limsup_{n\to\infty}\frac{1}{n}D\left(\rho_{A\bar{A}B\bar{B}}^{\otimes n}\middle\|\int_{-\infty}^{\infty} \textnormal{d}\beta_0(t)\left(\left(\mathcal{R}_{\bar{A}\to A\bar{A}}^{[t]}\otimes\mathcal{R}_{\bar{B}\to B\bar{B}}^{[t]}\right)(\rho_{\bar{A}\bar{B}})\right)^{\otimes n}\right) \label{eq:cemi-recovIII}
\end{align}
then also follows from the asymptotic achievability of the measured entropy (Lemma \ref{lem:asymptotic-achiev}). As for squashed entanglement, it is unclear how these recoverability lower bounds would directly imply faithfulness bounds.

Nevertheless, using again multivariate trace inequalities, a strengthened lower bound in terms of the measurement set $\nsf{SEP}(A:B)$ can be shown\,---\,compared to $\nsf{LOCC}_1(B\,;A)$ for squashed entanglement.

\begin{theorem}\label{thm:cemi-main}
Let $\rho_{A\bar{A}B\bar{B}}$ be any four-party state. We have
\begin{align}\label{eq:cemi-mainI}
    I(A|\bar{A}:B|\bar{B})_\rho\geq E_{\nsf{SEP}}(A:B)_{\rho}-\Big\{E(\bar{A}:\bar{B})_{\rho}-E_{\nsf{ALL}}(\bar{A}:\bar{B})_{\rho}\Big\}\,,
\end{align}
and consequently
\begin{align}\label{eq:cemi-mainII}
    I_{\nsf{CEMI}}(A:B)_\rho\geq\frac{1}{2}E_{\nsf{SEP}}^{\infty}(A:B)_\rho\geq\frac{1}{2}E_{\nsf{SEP}}(A:B)_\rho\,.
\end{align}
\end{theorem}

The stronger single-copy version in Eq.~\eqref{eq:cemi-mainI} is novel. The consequence in Eq.~\eqref{eq:cemi-mainI} corresponds to a strengthening of \cite[Equation 41]{Li14} that stated the (a priori weaker) lower bound with respect to $\nsf{LOCC}(A:B)$. One further advantage of our formulation in Theorem $\ref{thm:cemi-main}$ is that we have some information on the structure of the optimizer in lower bound $E^{\infty}_{\nsf{SEP}}(A:B)_{\rho}$, as in fact (see the proof of Theorem \ref{thm:cemi-main})
\begin{align}\label{eq:cemi-main-refined}
    I(A|\bar{A}:B|\bar{B})_\rho \geq
    &\;\frac{1}{n}D_{\nsf{SEP}}\left(\rho_{AB}^{\otimes n}\middle\|\int_{-\infty}^{\infty} \textnormal{d}\beta_0(t)\tr_{\bar{A}^n\bar{B}^n}\left[\left(\mathcal{R}_{\bar{A}\to A\bar{A}}^{[t]}\otimes\mathcal{R}_{\bar{B}\to B\bar{B}}^{[t]}\right)^{\otimes n}\left(\sigma_{\bar{A}^n:\bar{B}^n}\right)\right]\right)\nonumber\\
    & -\frac{1}{n}\log|\mathrm{poly}(n)|
\end{align}
for any separable state optimizer $\sigma_{\bar{A}^n:\bar{B}^n}\in\argmin_{\sigma_{\bar{A}^n\bar{B}^n}\in \nsf{Sep}(\bar{A}^n:\bar{B}^n)} \ D(\rho_{\bar{A}\bar{B}}^{\otimes n} \| \sigma_{\bar{A}^n\bar{B}^n})$, with local quantum channels
\begin{align}
\text{$\mathcal{R}_{\bar{A}\to A\bar{A}}^{[t]}(\cdot):=\left(\rho_{A\bar{A}}^{\frac{1+it}{2}}\rho_{\bar{A}}^{\frac{-1+it}{2}}\right)(\cdot)\left(\rho_{\bar{A}}^{\frac{-1-it}{2}}\rho_{A\bar{A}}^{\frac{1-it}{2}}\right)$ and $\mathcal{R}_{\bar{B}\to B\bar{B}}^{[t]}(\cdot)$ similar,}
\end{align}
and the probability density $\beta_0(t)=\frac{\pi}{2}\left(\cosh(\pi t)+1\right)^{-1}$. However, similarly as for squashed entanglement, this structure does not seem to further translate to the single-copy lower bound $E_{\nsf{SEP}}(A:B)_{\rho}$. Further lower bounds in terms of restricted fidelity and Schatten one-norm as leading to Eq.~\eqref{eq:locc-one-norm} are possible \cite{matthews09,Lami18}.

\begin{proof}[Proof of Theorem \ref{thm:cemi-main}]
The idea of the proof is similar as for Theorem \ref{thm:sq-main} and we first prove the bound in Eq.~\eqref{eq:cemi-mainI}. Namely, for a separable state optimizer
    \begin{align}
        \sigma_{\bar{A}:\bar{B}} \in \argmin_{\sigma_{\bar{A}\bar{B}}\in \nsf{Sep}(\bar{A}:\bar{B})} \ D(\rho_{\bar{A}\bar{B}} \| \sigma_{\bar{A}\bar{B}}), \qquad 
        \sigma_{\bar{A}:\bar{B}} = \sum_k\sigma^k_{\bar{A}}\otimes\sigma^k_{\bar{B}} \,,
    \end{align}
    using the four matrix Golden-Thompson inequality for the Schatten two-norm from \cite[Corollary 3.3]{sutter16} and the variational characterization from Lemma \ref{lem:variationalbound} with the choice $\omega_{A\bar{A}B\bar{B}}=X_{A:B}\otimes Y_{\bar{A}\bar{B}}$ with general $Y_{AB}>0$ and $X_{A:B}\in\nsf{SEP}(A:B)$ to be optimized over, we find
    \begin{align}
        &I(A|\bar{A}:B|\bar{B})_\rho+E(\bar{A}:\bar{B})_{\rho}\nonumber\\
        & =D\left(\rho_{A\bar{A}B\bar{B}}\middle\|\exp\left(\log\rho_{A\bar{A}}\otimes\rho_{B\bar{B}}+\log\sigma_{\bar{A}:\bar{B}}-\log\rho_{\bar{A}}\otimes\rho_{\bar{B}}\right)\right)\\
        &=\sup_{\omega_{A\bar{A}B\bar{B}} > 0}\Big\{\tr\left[\rho_{A\bar{A}B\bar{B}}\log\omega_{A\bar{A}B\bar{B}}\right]\notag\\
        &\qquad - \log\tr\left[\exp\left(\log\omega_{A\bar{A}B\bar{B}}+\log\rho_{A\bar{A}}\otimes\rho_{B\bar{B}} +\log\sigma_{\bar{A}:\bar{B}}-\log\rho_{\bar{A}}\otimes\rho_{\bar{B}}\right)\right]\Big\}\\
        &\geq \sup_{X_{A:B},Y_{\bar{A}\bar{B}} > 0}\Big\{\tr\left[\rho_{A\bar{A}B\bar{B}}\log X_{A:B}\otimes Y_{\bar{A}\bar{B}}\right]\nonumber\\
        &\quad-\log\tr\left[\exp\left(\log X_{A:B}\otimes Y_{\bar{A}\bar{B}}+\log\rho_{A\bar{A}}\otimes\rho_{B\bar{B}}+\log\sigma_{\bar{A}:\bar{B}}-\log\rho_{\bar{A}}\otimes\rho_{\bar{B}}\right)\right]\Big\}\\
        &\geq \sup_{X_{A:B},Y_{\bar{A}\bar{B}} > 0}\Bigg\{\tr\left[\rho_{A\bar{A}B\bar{B}}\log X_{A:B}\otimes Y_{\bar{A}\bar{B}}\right]\nonumber\\
        &\quad-\log\tr\Big[\left(X_{A:B}\otimes Y_{\bar{A}\bar{B}}\right)\underbrace{\int_{-\infty}^{\infty} \textnormal{d}\beta_0(t)\sum_k\mathcal{R}_{\bar{A}\to A\bar{A}}^{[t]}(\sigma_{\bar{A}}^k)\otimes\mathcal{R}_{\bar{B}\to B\bar{B}}^{[t]}(\sigma_{\bar{B}}^k)}_{=:\;\gamma_{A\bar{A}:B\bar{B}}} \Big]\Bigg\}\\
        &=\sup_{X_{A:B},Y_{\bar{A}\bar{B}} > 0}\Big\{\tr\left[\rho_{AB}\log X_{A:B}\right]+\tr\left[\rho_{\bar{A}\bar{B}}\log Y_{\bar{A}\bar{B}}\right]-\log\big(\tr\left[X_{A:B}\gamma_{A:B}\right]\cdot\tr\left[Y_{\bar{A}\bar{B}}\hat{\gamma}_{\bar{A}:\bar{B}}\right]\big)\Big\}\\
        &\geq E_{\nsf{SEP}}(A:B)_{\rho}+E_{\nsf{ALL}}(\bar{A}:\bar{B})_{\rho}\,,
    \end{align}
    where we set
    \begin{align}
        \text{$\hat{\gamma}_{\bar{A}:\bar{B}}:=\frac{\tr_{AB}\left[X_{A:B}\gamma_{A\bar{A}:B\bar{B}}\right]}{\tr\left[X_{A:B}\gamma_{A:B}\right]}$ with $X_{A:B}\in\nsf{SEP}(A:B)$,}
    \end{align}
    and used that $\gamma_{A:B}\in\nsf{SEP}(A:B)$ as well as $\hat{\gamma}_{\bar{A}:\bar{B}}\in\nsf{SEP}(\bar{A}:\bar{B})$ inherit the separability structure from the choice $X_{A:B}\in\nsf{SEP}(A:B)$.

    Next, the first step in Eq.~\eqref{eq:cemi-mainII} follows from the additivity of $I(A|\bar{A}:B|\bar{B})$ on tensor product states together with the asymptotic achievability of the measured relative entropy in Lemma \ref{lem:asymptotic-achiev}.

    Finally, the second step in Eq.~\eqref{eq:cemi-mainII} can be deduced from the super-additivity result \cite[Theorem 1]{piani09}
    \begin{align}
        E_{\nsf{SEP}}(A_1A_2:B_1B_2)_{\rho}\geq E_{\nsf{SEP}}(A_1:B_1)_{\rho}+E_{\nsf{SEP}}(A_2:B_2)_{\rho}\,,
    \end{align}
    noting that\,---\,in the notation of \cite{piani09}\,---\,the set of measurements $\nsf{SEP}(A:B)$ is compatible with the set of states $\nsf{Sep}(A:B)$.
\end{proof}

Alternatively, we can derive PPT bounds, where the set of measurements and the set of states are both in terms of PPT. We are not aware of any previous such bounds in the literature.

\begin{proposition}\label{prop:cemi-ppt}
Let $\rho_{A\bar{A}B\bar{B}}$ be any four-party state. We have
\begin{align}\label{eq:cemi-pptI}
    I(A|\bar{A}:B|\bar{B})_\rho\geq P_{\nsf{PPT}}(A:B)_{\rho}-\Big\{P(\bar{A}:\bar{B})_{\rho}-P_{\nsf{ALL}}(\bar{A}:\bar{B})_{\rho}\Big\}\,,
\end{align}
and consequently
\begin{align}\label{eq:cemi-pptII}
    I_{\nsf{CEMI}}(A:B)_\rho\geq\frac{1}{2}P_{\nsf{PPT}}^{\infty}(A:B)_\rho\geq\frac{1}{2}P_{\nsf{PPT}}(A:B)_\rho\,.
\end{align}
\end{proposition}

Note that the lower bounds in Proposition \ref{prop:cemi-ppt} are in general not directly comparable to the bounds from Theorem \ref{thm:cemi-main}, as both the set of measurements as well as the set of states is enlarged. Moreover, the same form as in Eq.~\eqref{eq:cemi-main-refined} is available and lower bounds in terms of restricted fidelity and Schatten one-norm as leading to Eq.~\eqref{eq:locc-one-norm} are possible as well \cite{matthews09,Lami18}.

\begin{proof}[Proof of Proposition \ref{prop:cemi-ppt}]
    The first part of the proof is similar as that of Theorem \ref{thm:cemi-main}. Namely, for a PPT state optimizer
    \begin{align}
        \sigma_{\bar{A}\bar{B}}\in\argmin_{\sigma_{\bar{A}\bar{B}}\in \nsf{ppt}(\bar{A}:\bar{B})} \ D(\rho_{\bar{A}\bar{B}} \| \sigma_{\bar{A}\bar{B}})
    \end{align}
    we find
    \begin{align}
        &I(A|\bar{A}:B|\bar{B})_\rho+P(\bar{A}:\bar{B})_{\rho}\nonumber\\
        &\geq \sup_{X_{AB},Y_{\bar{A}\bar{B}} > 0}\Bigg\{\tr\left[\rho_{A\bar{A}B\bar{B}}\log X_{AB}\otimes Y_{\bar{A}\bar{B}}\right]\nonumber\\
        &\quad-\log\tr\Big[\left(X_{AB}\otimes Y_{\bar{A}\bar{B}}\right)\underbrace{\int_{-\infty}^{\infty} \textnormal{d}\beta_0(t)\left(\mathcal{R}_{\bar{A}\to A\bar{A}}^{[t]}\otimes\mathcal{R}_{\bar{B}\to B\bar{B}}^{[t]}\right)(\sigma_{\bar{A}\bar{B}})}_{=:\;\gamma_{A\bar{A}B\bar{B}}} \Big]\Bigg\}\\
        &=\sup_{X_{AB},Y_{\bar{A}\bar{B}} > 0}\Bigg\{\tr\left[\rho_{AB}\log X_{AB}\right]+\tr\left[\rho_{\bar{A}\bar{B}}\log Y_{\bar{A}\bar{B}}\right]-\log\left(\tr\left[X_{AB}\gamma_{AB}\right]\cdot\tr\left[Y_{\bar{A}\bar{B}}\hat{\gamma}_{\bar{A}\bar{B}}\right]\right)\Bigg\}\,,\label{eq:PPT-laststep}
    \end{align}
    where we set
    \begin{align}
        \text{$\hat{\gamma}_{\bar{A}\bar{B}}:=\frac{\tr_{AB}\left[X_{AB}\gamma_{A\bar{A}B\bar{B}}\right]}{\tr\left[X_{AB}\gamma_{AB}\right]}$ for the choice $X_{AB}\in\nsf{PPT}(A:B)$.}
    \end{align}
    Eq.~\eqref{eq:PPT-laststep} is then further lower bounded to the claimed inequality
    \begin{align}
        I(A|\bar{A}:B|\bar{B})_\rho+P(\bar{A}:\bar{B})_{\rho}\geq P_{\nsf{PPT}}(A:B)_{\rho}+P_{\nsf{ALL}}(\bar{A}:\bar{B})_{\rho}
    \end{align}
    once it is realized that both $\gamma_{AB}\in\nsf{ppt}(A:B)$ and $\hat{\gamma}_{\bar{A}\bar{B}}\in\nsf{ppt}(\bar{A}:\bar{B})$ inherit the PPT structure. This follows as by inspection
    \begin{align}
        &T_{B\bar{B}}\circ\mathcal{R}_{\bar{B}\to B\bar{B}}^{[t]}=\widetilde{\mathcal{R}}_{\bar{B}\to B\bar{B}}^{[t]}\circ T_{\bar{B}}\\
        &\text{for}\quad\widetilde{\mathcal{R}}_{\bar{B}\to B\bar{B}}^{[t]}(\cdot):=\left(\left(\rho_{B\bar{B}}^{T_{\bar{B}}}\right)^{\frac{1+it}{2}}\left(\rho_{\bar{B}}^{T}\right)^{\frac{-1+it}{2}}\right)(\cdot)^{T_{\bar{B}}}\left(\left(\rho_{\bar{B}}^T\right)^{\frac{-1-it}{2}}\left(\rho_{B\bar{B}}^{T_{\bar{B}}}\right)^{\frac{1-it}{2}}\right)
    \end{align}
    and hence $\gamma_{A\bar{A}B\bar{B}}\in\nsf{ppt}(A\bar{A}:B\bar{B})$, and further
    \begin{align}
        \left(\tr_{AB}\left[X_{AB}\gamma_{A\bar{A}B\bar{B}}\right]\right)^{T_{\bar{B}}}&=\tr_{AB}\left[X_{AB}\gamma_{A\bar{A}B\bar{B}}^{T_{\bar{B}}}\right]=\tr_{AB}\left[\left(X_{AB}\gamma_{A\bar{A}B\bar{B}}^{T_{\bar{B}}}\right)^{T_B}\right]\\
        &=\tr_{AB}\left[\left(X_{AB}\right)^{T_B}\gamma_{A\bar{A}B\bar{B}}^{T_{B\bar{B}}}\right]=\tr_{AB}\left[X_{AB}\gamma_{A\bar{A}B\bar{B}}\right]\,.
    \end{align}

    Finally, Eq.\eqref{eq:cemi-pptII} follows as in the proof of Theorem \ref{thm:cemi-main}, except now using \cite[Theorem 1]{piani09}
    \begin{align}
        P_{\nsf{PPT}}(A_1A_2:B_1B_2)_{\rho}\geq P_{\nsf{PPT}}(A_1:B_1)_{\rho}+P_{\nsf{PPT}}(A_2:B_2)_{\rho}\,,
    \end{align}
    noting that\,---\,in the notation of \cite{piani09}\,---\,the set of measurements $\nsf{PPT}(A:B)$ is compatible with the set of states $\nsf{ppt}(A:B)$.
\end{proof}


\subsection{Piani based relative entropy of entanglement}
\label{sec:piani}

The previously known CEMI lower bound proof proceeded via two steps of multipartite monogamy inequalities \cite{Li14} (see also the alternative \cite{wilde14}), going through the relative entropy of entanglement and prominently making use of Piani's results \cite{piani09}. As the intermediate steps of these proofs are of independent interest, we now give simple and direct proofs for strengthened single-copy versions of these steps. The first bound is as follows.

\begin{proposition}\label{prop:piani-cemi}
    Let $\rho_{AB\bar{A}\bar{B}}$ be any four-party state. We have
    \begin{align}\label{eq:state-merging}
        I(A|\bar{A}:B|\bar{B})_\rho\geq E_{\nsf{ALL}}(A\bar{A}:B\bar{B})_\rho-E(\bar{A}:\bar{B})_\rho\,,
    \end{align}
    and consequently
    \begin{align}\label{eq:cemi-infty}
        I(A|\bar{A}:B|\bar{B})_\rho\geq E^{\infty}(A\bar{A}:B\bar{B})_\rho-E^{\infty}(\bar{A}:\bar{B})_\rho\,.
    \end{align}
\end{proposition}

We note that the stronger single-copy version in Eq.~\eqref{eq:state-merging} is novel. The consequence in Eq.~\eqref{eq:cemi-infty} is \cite[Equation 40]{Li14}, which was based on the asymptotic achievability of partial state merging \cite{Yang08} together with the asymptotic continuity \cite{Donald99,Radtke06} and non-lockability \cite{Horodecki05-2} of the relative entropy of entanglement. In contrast, our proof is elementary via matrix trace inequalities.

\begin{proof}
The proof is a simplified version of the arguments leading to Theorem \ref{thm:cemi-main}. We only sketch the steps: For a separable state optimizer
				\begin{align}
        \sigma_{\bar{A}:\bar{B}} \in \argmin_{\sigma_{\bar{A}\bar{B}}\in \nsf{Sep}(\bar{A}:\bar{B})} \ D(\rho_{\bar{A}\bar{B}} \| \sigma_{\bar{A}\bar{B}}), \qquad 
        \sigma_{\bar{A}:\bar{B}} = \sum_k\sigma^k_{\bar{A}}\otimes\sigma^k_{\bar{B}} \,,
    \end{align}
    we estimate for Eq.~\eqref{eq:state-merging} that 
    \begin{align}
        &I(A|\bar{A}:B|\bar{B})_\rho+E(\bar{A}:\bar{B})_\rho\nonumber\\
        &=D(\rho_{A\bar{A}B\bar{B}}\|\exp(\log\sigma_{\bar{A}:\bar{B}}+\log\rho_{A\bar A}\otimes\rho_{B\bar B}-\log\rho_{\bar A}\otimes\rho_{\bar B}))\\
        &\geq\sup_{\omega_{A\bar{A}B\bar{B}}>0}\Bigg\{\tr[\rho_{A\bar{A}B\bar{B}}\log\omega_{A\bar{A}B\bar{B}}]-\log\tr\Big[\omega_{A\bar{A}B\bar{B}}\int_{-\infty}^{\infty} \textnormal{d}\beta_0(t)\sum_k\mathcal{R}_{\bar{A}\to A\bar{A}}^{[t]}(\sigma^k_{\bar{A}})\otimes\mathcal{R}_{\bar{B}\to B\bar{B}}^{[t]}(\sigma^k_{\bar{B}})\Big]\Bigg\}\\
        &\geq E_{\nsf{ALL}}(A\bar{A}:B\bar{B})_\rho\,,
    \end{align}
    with local quantum channels
    \begin{align}
        \text{$\mathcal{R}_{\bar{A}\to A\bar{A}}^{[t]}(\cdot):=\left(\rho_{A\bar{A}}^{\frac{1+it}{2}}\rho_{\bar{A}}^{\frac{-1+it}{2}}\right)(\cdot)\left(\rho_{\bar{A}}^{\frac{-1-it}{2}}\rho_{A\bar{A}}^{\frac{1-it}{2}}\right)$ and $\mathcal{R}_{\bar{B}\to B\bar{B}}^{[t]}(\cdot)$ similar,}
    \end{align}
    and the probability density $\beta_0(t)=\frac{\pi}{2}\left(\cosh(\pi t)+1\right)^{-1}$. Eq.~\eqref{eq:cemi-infty} then follows by the additivity of $I(A|\bar{A}:B|\bar{B})$ on tensor product states together with the asymptotic achievability of the measured relative entropy from Lemma \ref{lem:asymptotic-achiev}.
\end{proof}

Having Proposition \ref{prop:piani-cemi} at hand, we can employ \cite[Theorem 1]{piani09} in the form
\begin{align}\label{eq:piani-sep}
    E(A\bar{A}:B\bar{B})_\rho\geq E_{\nsf{SEP}}(A:B)_\rho+E(\bar{A}:\bar{B})_\rho
\end{align}
to again conclude that
\begin{align}
    I(A|\bar{A}:B|\bar{B})_\rho\geq E_{\nsf{SEP}}^{\infty}(A:B)_\rho\geq E_{\nsf{SEP}}(A:B)_\rho\,,
\end{align}
as proven directly in Theorem \ref{thm:cemi-main}.

Finally, one can equally show that
\begin{align}
    I(A|\bar{A}:B|\bar{B})_\rho\geq P_{\nsf{ALL}}(A\bar{A}:B\bar{B})_\rho-P(\bar{A}:\bar{B})_\rho
\end{align}
and then use the PPT bound
\begin{align}\label{eq:piani-ppt}
    P(A\bar{A}:B\bar{B})_\rho\geq P_{\nsf{PPT}}(A:B)_\rho+P(\bar{A}:\bar{B})_\rho
\end{align}
from \cite[Theorem 1]{piani09} to arrive at the second part of Proposition \ref{prop:cemi-ppt}.


\subsection{Multipartite extensions}
\label{sec:multipartite}

Our results can (partly) be extended to multipartite settings. In the following, we discuss tripartite systems\,---\,with straightforward generalization to more parties. The quantity
\begin{align}
    &I(A|\bar{A}:B|\bar{B}:C|\bar{C})_\rho:=I(A\bar{A}:B\bar{B}:C\bar{C})_\rho-I(\bar{A}:\bar{B}:\bar{C})_\rho\\
        &\text{with}\quad I(A:B:C):=H(A)_\rho+H(B)_\rho+H(C)_\rho-H(ABC)_\rho
\end{align}
gives rise to the tripartite CEMI as \cite{Yang08}
\begin{align}
I_{\nsf{CEMI}}(A:B:C)_\rho:=\frac{1}{2}\inf_{\rho_{A\bar{A}B\bar{B}\bar{C}}}I(A|\bar{A}:B|\bar{B}:C|\bar{C})_\rho\,,
\end{align}
where the infimum goes over all tripartite extensions $\rho_{A\bar{A}B\bar{B}C\bar{C}}$ of $\rho_{ABC}$ on systems $\bar{A}\bar{B}\bar{C}$ (with no bound on the dimensions of $\bar{A}$, $\bar{B}$, $\bar{C}$). We first note the following recoverability lower bounds that resolve a conjecture from \cite{wilde14}.

\begin{proposition}\label{prop:cemi-recov-multi}
Let $\rho_{A\bar{A}B\bar{B}C\bar{C}}$ be any six-party state. We have
\begin{align}
    I(A|\bar{A}:B|\bar{B}:C|\bar{C})_\rho & \geq D_{\nsf{ALL}}\left(\rho_{A\bar{A}B\bar{B}C\bar{C}}\middle\|\int_{-\infty}^{\infty} \textnormal{d}\beta_0(t)\left(\mathcal{R}_{\bar{A}\to A\bar{A}}^{[t]}\otimes\mathcal{R}_{\bar{B}\to B\bar{B}}^{[t]}\otimes\mathcal{R}_{\bar{C}\to C\bar{C}}^{[t]}\right)(\rho_{\bar{A}\bar{B}\bar{C}})\right)\\
    I(A|\bar{A}:B|\bar{B}:C|\bar{C})_\rho & \geq -\int_{-\infty}^{\infty} \textnormal{d}\beta_0(t)\log F\left(\rho_{A\bar{A}B\bar{B}C\bar{C}},\left(\mathcal{R}_{\bar{A}\to A\bar{A}}^{[t]}\otimes\mathcal{R}_{\bar{B}\to B\bar{B}}^{[t]}\otimes\mathcal{R}_{\bar{C}\to C\bar{C}}^{[t]}\right)(\rho_{\bar{A}\bar{B}\bar{C}})\right) ,
\end{align}
with local quantum channels
\begin{align}
    \text{$\mathcal{R}_{\bar{A}\to A\bar{A}}^{[t]}(\cdot):=\left(\rho_{A\bar{A}}^{\frac{1+it}{2}}\rho_{\bar{A}}^{\frac{-1+it}{2}}\right)(\cdot)\left(\rho_{\bar{A}}^{\frac{-1-it}{2}}\rho_{A\bar{A}}^{\frac{1-it}{2}}\right)$ and $\mathcal{R}_{\bar{B}\to B\bar{B}}^{[t]}(\cdot),\mathcal{R}_{\bar{C}\to C\bar{C}}^{[t]}(\cdot)$ similar,}
\end{align}
and the probability density $\beta_0(t)=\frac{\pi}{2}\left(\cosh(\pi t)+1\right)^{-1}$.
\end{proposition}
As the additional, third recovery map $\mathcal{R}_{\bar{C}\to C\bar{C}}^{[t]}$ commutes with the other tensor product recovery maps, the proof is exactly the same as the proof in the bipartite case (Proposition \ref{prop:cemi-recover}). Additionally, the corresponding regularized lower bound as
\begin{align}
\limsup_{n\to\infty}\frac{1}{n}D\left(\rho_{A\bar{A}B\bar{B}C\bar{C}}^{\otimes n}\middle\|\int_{-\infty}^{\infty} \textnormal{d}\beta_0(t)\left(\left(\mathcal{R}_{\bar{A}\to A\bar{A}}^{[t]}\otimes\mathcal{R}_{\bar{B}\to B\bar{B}}^{[t]}\otimes\mathcal{R}_{\bar{C}\to C\bar{C}}^{[t]}\right)(\rho_{\bar{A}\bar{B}\bar{C}})\right)^{\otimes n}\right)
\end{align}
then also follows from the asymptotic achievability of the measured entropy (Lemma \ref{lem:asymptotic-achiev}).

We find the following faithfulness bound in terms of tripartite separability.

\begin{proposition}\label{thm:cemi-multi}
Let $\rho_{A\bar{A}B\bar{B}C|\bar{C}}$ be any six-party state. We have
\begin{align}
    I(A|\bar{A}:B|\bar{B}:C|\bar{C})_\rho&\geq E_{\nsf{SEP}}(A:B:C)_{\rho}-\Big\{E(\bar{A}:\bar{B}:\bar{C})_{\rho}-E_{\nsf{ALL}}(\bar{A}:\bar{B}:\bar{C})_{\rho}\Big\}\label{eq:cemi-multiI}\\
    I(A|\bar{A}:B|\bar{B}:C|\bar{C})_\rho&\geq P_{\nsf{PPT}}(A:B:C)_{\rho}-\Big\{P(\bar{A}:\bar{B}:\bar{C})_{\rho}-P_{\nsf{ALL}}(\bar{A}:\bar{B}:\bar{C})_{\rho}\Big\}\,,\label{eq:cemi-multiII}
\end{align}
and consequently
\begin{align}
    I_{\nsf{CEMI}}(A:B:C)_\rho&\geq\frac{1}{2}E_{\nsf{SEP}}^{\infty}(A:B:C)_\rho\geq\frac{1}{2}E_{\nsf{SEP}}(A:B:C)_\rho\label{eq:cemi-multiIII}\\
    I_{\nsf{CEMI}}(A:B:C)_\rho&\geq\frac{1}{2}P_{\nsf{PPT}}^{\infty}(A:B:C)_\rho\geq\frac{1}{2}P_{\nsf{PPT}}(A:B:C)_\rho\,.\label{eq:cemi-multiIV}
\end{align}
\end{proposition}

This strengthens the conceptually different multipartite CEMI faithfulness bounds from \cite{wilde14}. Further lower bounds in terms of restricted fidelity and Schatten one-norm as leading to Eq.~\eqref{eq:locc-one-norm} are possible \cite{Lancien13}. The proof is similar as in the respective bipartite cases, Theorem \ref{thm:cemi-main} and Proposition \ref{prop:cemi-ppt}, and is given in Appendix \ref{app:missing-proofs}. 

The tripartite squashed entanglement is defined as \cite{Yang09,Avis08}
\begin{align}
    I_{\nsf{SQ}}(A_1:A_2:A_3)&:=\inf_{\rho_{A_1A_2A_3}}I(A_1:A_2:A_3|C)_{\rho}\\
    \text{via the tripartite CQMI}\;I(A_1:A_2:A_3|C)_{\rho}&:=I(A_1:A_2|C)_{\rho}+I(A_1A_2:A_3|C)_{\rho}\,,
\end{align}
where the infimum is over all four-party quantum state extensions $\rho_{A_1A_2A_3C}$ on any system $C$ (with no bound on the dimension of $C$).\footnote{In reference \cite{Yang09} other definitions of multipartite squashed entanglement are explored as well, which we do not discuss here, but should be amenable to similar considerations.} We note the following recoverability lower bounds.

\begin{proposition}\label{prop:squashed-recov-multi}
Let $\rho_{A_1A_2A_3C}$ be any four-party state. We have
\begin{align}
    I(A_1:A_2:A_3|C)_{\rho} & \geq D_{\nsf{ALL}}\left(\rho_{A_1A_2A_3C}\middle\|\int_{-\infty}^{\infty} \textnormal{d}\beta_0(t)\left(\mathcal{R}_{C\to A_3C}^{[t]}\circ\mathcal{R}_{C\to A_2C}^{[t]}\right)(\rho_{A_1C})\right)\label{eq:cqmi-3II}\\
    I(A_1:A_2:A_3|C)_{\rho} & \geq -\int_{-\infty}^{\infty} \textnormal{d}\beta_0(t)\log F\left(\rho_{A_1A_2A_3C},\left(\mathcal{R}_{C\to A_3C}^{[t]}\circ\mathcal{R}_{C\to A_2C}^{[t]}\right)(\rho_{A_1C})\right)\,,\label{eq:cqmi-3III}
\end{align}
with quantum channels
\begin{align}
    \text{$\mathcal{R}_{\bar{C}\to A_2C}^{[t]}(\cdot):=\left(\rho_{A_2C}^{\frac{1+it}{2}}\rho_{C}^{\frac{-1+it}{2}}\right)(\cdot)\left(\rho_{C}^{\frac{-1-it}{2}}\rho_{A_2C}^{\frac{1-it}{2}}\right)$ and $\mathcal{R}_{\bar{C}\to A_3C}^{[t]}(\cdot)$ similar,}
\end{align}
and the probability density $\beta_0(t)=\frac{\pi}{2}\left(\cosh(\pi t)+1\right)^{-1}$. By the symmetry of $I(A_1:A_2:A_3|C)_{\rho}$ under $A_1\leftrightarrow A_2\leftrightarrow A_3$ other orderings are possible as well.
\end{proposition}
The proof is as in \cite{sutter16,Sutter17} via multivariate trace inequalities and is given in Appendix \ref{app:missing-proofs}. Additionally, the corresponding regularized lower bound as
\begin{align}
    \limsup_{n\to\infty}\frac{1}{n}D\left(\rho_{A_1A_2A_3C}^{\otimes n}\middle\|\int_{-\infty}^{\infty} \textnormal{d}\beta_0(t)\left(\left(\mathcal{R}_{C\to A_3C}^{[t]}\circ\mathcal{R}_{C\to A_2C}^{[t]}\right)(\rho_{A_1C})\right)^{\otimes n}\right)\label{eq:cqmi-3I}
\end{align}
then also follows from the asymptotic achievability of the measured entropy (Lemma \ref{lem:asymptotic-achiev}).

However, we do not know how to show faithfulness lower bounds of $I_{\nsf{SQ}}(A_1:A_2:A_3)$ with respect to global separability $\nsf{SEP}(A_1:A_2:A_3)$. As also noted in \cite{liwinter14,wilde14}, this difficulty arises because compared to the multipartite CEMI case, there is now only one extension system $C$ that all operators act on.


\section{Outlook}
\label{sec:outlook}

In addition to exploring applications of our variational formulas for quantum relative entropy under restricted measurements, there are two immediate questions that remain open around entanglement monogamy inequalities in the spirit of this manuscript. First, is multipartite squashed entanglement faithful? Second, and as an extension of the separability refinements of SSA, is there a connection between the quantum conditional mutual information and exact quantum Markov chains \cite{hayden04,linden08}? We hope that our direct matrix analysis approach can shine some further light on these questions. Lastly, it would also be interesting to explore applications of the CEMI entanglement measure and its characterizations.



\section*{Acknowledgments}

We thank Ludovico Lami for feedback. MB acknowledges funding by the European Research Council (ERC Grant Agreement No.~948139). MT is supported by the National Research Foundation, Singapore and A*STAR under its CQT Bridging Grant and by NUS startup grants (R-263-000-E32-133 and R-263-000-E32-731).


\appendix

\section{Deferred proofs}\label{app:missing-proofs}

\begin{proof}[Proof of Proposition \ref{prop:cemi-recover}]
We first prove Eq.~\eqref{eq:cemi-recovI} by writing
\begin{align}
&I(A|\bar A:B\bar B)_\rho\nonumber\\
&\quad= D\left(\rho_{A\bar AB\bar B}\middle\|\exp(\log\rho_{\bar A\bar B}+\log\rho_{A\bar A}\otimes\rho_{B\bar B}-\log\rho_{\bar A}\otimes\rho_{\bar B})\right)\\
&\quad= \sup_{\omega_{A\bar AB\bar B} > 0}\Big\{\tr\left[\rho_{A\bar AB\bar B}\log\omega_{A\bar AB\bar B}\right] \nonumber\\
&\quad \qquad - \log \tr\left[\exp (\log\rho_{\bar A\bar B}+\log\rho_{A\bar A}\otimes\rho_{B\bar B}-\log\rho_{\bar A}\otimes\rho_{\bar B}+ \log\omega_{A\bar AB\bar B}\right]\Big\}\\
&\quad \geq\sup_{\omega_{A\bar AB\bar B} > 0}\Bigg\{\tr\left[\rho_{A\bar AB\bar B}\log\omega_{A\bar AB\bar B}\right]- \log \tr\left[\omega_{A\bar AB\bar B}\int_{-\infty}^{\infty} \textnormal{d}\beta_0(t)\left(\mathcal{R}_{\bar{A}\to A\bar{A}}^{[t]}\otimes\mathcal{R}_{\bar{B}\to B\bar{B}}^{[t]}\right)(\rho_{\bar{A}\bar{B}})\right]\Bigg\}\\
&\quad= D_{\nsf{ALL}}\left(\rho_{A\bar{A}B\bar{B}}\middle\|\int_{-\infty}^{\infty} \textnormal{d}\beta_0(t)\left(\mathcal{R}_{\bar{A}\to A\bar{A}}^{[t]}\otimes\mathcal{R}_{\bar{B}\to B\bar{B}}^{[t]}\right)(\rho_{\bar{A}\bar{B}})\right)\,,\label{eq:ALL-firststepI}
\end{align}
where we employed the four matrix Golden-Thompson inequality for the Schatten two-norm from \cite[Corollary 3.3]{sutter16}. Next, Eq.~\eqref{eq:cemi-recovIII} directly follows from Eq.~\eqref{eq:ALL-firststepI} via the additivity of CEMI on tensor product states, together with the asymptotic achievability of the measured relative entropy in Lemma \ref{lem:asymptotic-achiev}. Finally, for the proof of Eq.~\eqref{eq:cemi-recovII}, we follow the same consideration as in ~\cite[Appendix F]{sutter16}. Namely, the Peierls-Bogoliubov inequality for Hermitian matrices $G_1$ and $G_2$ with $\tr[\exp(G_1)]=1$ states that
\begin{align}\label{eq:PB}
-\tr[G_2\exp(G_1)]\geq-\log\tr[\exp(G_1+G_2)]\,.
\end{align}
Choosing $G_1=\log\rho_{A\bar{A}B\bar{B}}$ and
$G_2=\frac{1}{2}\left(-\log\rho_{A\bar AB\bar B}+\log\rho_{A\bar A}\otimes\rho_{B\bar B}-\log\rho_{\bar A\bar B}+\log\rho_{\bar A}\otimes\rho_{\bar B}\right)$
yields
\begin{align}
&I(A|\bar A:B\bar B)_\rho\nonumber\\
&\quad \geq-2\log\tr\left[\exp\left(\frac{1}{2}(-\log\rho_{A\bar AB\bar B}+\log\rho_{A\bar A}\otimes\rho_{B\bar B}-\log\rho_{\bar A\bar B}+\log\rho_{\bar A}\otimes\rho_{\bar B})\right)\right]\\
&\quad\geq-2\int_{-\infty}^{\infty} \textnormal{d}\beta_0(t)\log\left\|\rho_{A\bar{A}B\bar{B}}^{\frac{1+it}{2}}\left(\rho_{A\bar A}\otimes\rho_{B\bar B}\right)^{\frac{1+it}{2}}\left(\rho_{\bar A}\otimes\rho_{\bar B}\right)^{-\frac{1+it}{2}}\rho_{\bar A\bar B}^{\frac{1+it}{2}}\right\|_1\\
&\quad=-\int_{-\infty}^{\infty} \textnormal{d}\beta_0(t)\log F\left(\rho_{A\bar{A}B\bar{B}},\left(\mathcal{R}_{\bar{A}\to A\bar{A}}^{[t]}\otimes\mathcal{R}_{\bar{B}\to B\bar{B}}^{[t]}\right)(\rho_{\bar{A}\bar{B}})\right)\,,
\end{align}
where we subsequently employed the four matrix Golden-Thompson inequality for the Schatten one-norm from \cite[Corollary 3.3]{sutter16}.
\end{proof}

\begin{proof}[Proof of Proposition \ref{thm:cemi-multi}]
    We first prove Eq.\eqref{eq:cemi-multiI} and proceed as in the bipartite case of Theorem \ref{thm:cemi-main}. Namely, for a separable state optimizer
				\begin{align}
        \sigma_{\bar{A}:\bar{B}:\bar{C}} \in \argmin_{\sigma_{\bar{A}\bar{B}\bar{C}}\in \nsf{Sep}(\bar{A}:\bar{B}:\bar{C})} \ D(\rho_{\bar{A}\bar{B}\bar{C}} \| \sigma_{\bar{A}\bar{B}\bar{C}}), \qquad
        \sigma_{\bar{A}:\bar{B}:\bar{C}} = \sum_k\sigma^k_{\bar{A}}\otimes\sigma^k_{\bar{B}}\otimes\sigma^k_{\bar{C}} \,,
    \end{align}
    we find
    \begin{align}
        &I(A|\bar{A}:B|\bar{B}:C|\bar{C})_\rho+E(\bar{A}:\bar{B}:\bar{C})_{\rho}\nonumber\\
        &\ =D\left(\rho_{A\bar{A}B\bar{B}C\bar{C}}\middle\|\exp\left(\log\rho_{A\bar{A}}\otimes\rho_{B\bar{B}}\otimes\rho_{C\bar{C}}+\log\sigma_{\bar{A}:\bar{B}:\bar{C}}-\log\rho_{\bar{A}}\otimes\rho_{\bar{B}}\otimes\rho_{\bar{C}}\right)\right)\\
        &\ \geq \sup_{X_{A:B:C},Y_{\bar{A}\bar{B}\bar{C}} > 0}\Bigg\{\tr\left[\rho_{A\bar{A}B\bar{B}C\bar{C}}\log X_{A:B:C}\otimes Y_{\bar{A}\bar{B}\bar{C}}\right]\nonumber\\
        &\ \quad-\log\tr\Big[\left(X_{A:B:C}\otimes Y_{\bar{A}\bar{B}\bar{C}}\right)\underbrace{\int_{-\infty}^{\infty} \textnormal{d}\beta_0(t)\sum_k\mathcal{R}_{\bar{A}\to A\bar{A}}^{[t]}(\sigma_{\bar{A}}^k)\otimes\mathcal{R}_{\bar{B}\to B\bar{B}}^{[t]}(\sigma_{\bar{B}}^k)\otimes\mathcal{R}_{\bar{C}\to C\bar{C}}^{[t]}(\sigma_{\bar{C}}^k)}_{=:\;\gamma_{A\bar{A}:B\bar{B}:C\bar{C}}} \Big]\Bigg\}\\
        &\ =\sup_{X_{A:B:C},Y_{\bar{A}\bar{B}\bar{C}} > 0}\Big\{\tr\left[\rho_{ABC}\log X_{A:B:C}\right]+\tr\left[\rho_{\bar{A}\bar{B}\bar{C}}\log Y_{\bar{A}\bar{B}\bar{C}}\right] \nonumber\\
        &\ \quad -\log\big(\tr\left[X_{A:B:C}\gamma_{A:B:C}\right]\cdot\tr\left[Y_{\bar{A}\bar{B}\bar{C}}\hat{\gamma}_{\bar{A}:\bar{B}:\bar{C}}\right]\big)\Big\}\\
        &\ \geq E_{\nsf{SEP}}(A:B:C)_{\rho}+E_{\nsf{ALL}}(\bar{A}:\bar{B}:\bar{C})_{\rho}\,,
    \end{align}
    where we set
    \begin{align}
        \text{$\hat{\gamma}_{\bar{A}:\bar{B}:\bar{C}}:=\frac{\tr_{ABC}\left[X_{A:B:C}\gamma_{A\bar{A}:B\bar{B}:C\bar{C}}\right]}{\tr\left[X_{A:B:C}\gamma_{A:B:C}\right]}$ with $X_{A:B:C}\in\nsf{SEP}(A:B:C)$,}
    \end{align}
    and used that $\gamma_{A:B:C}\in\nsf{SEP}(A:B:C)$ as well as $\hat{\gamma}_{\bar{A}:\bar{B}:\bar{C}}\in\nsf{SEP}(\bar{A}:\bar{B}:\bar{C})$ inherit the separability structure from the choice $X_{A:B:C}\in\nsf{SEP}(A:B:C)$. The crucial point is that the recovery maps all commute and thus the argument is the same as in the bipartite case of Theorem \ref{thm:cemi-main}. Eq.~\eqref{eq:cemi-multiIII} then follows from the multipartite version of the super-additivity result \cite[Theorem 1]{piani09}
    \begin{align}
        E_{\nsf{SEP}}(A_1A_2:B_1B_2:C_1C_2)_{\rho}\geq E_{\nsf{SEP}}(A_1:B_1:C_1)_{\rho}+E_{\nsf{SEP}}(A_2:B_2:C_2)_{\rho}\,.
    \end{align}

    For the proof of Eq.~\eqref{eq:cemi-multiII}, we follow Proposition \ref{prop:cemi-ppt} for the bipartite case. Namely, for a PPT state optimizer
    \begin{align}
        \sigma_{\bar{A}\bar{B}\bar{C}}\in\argmin_{\sigma_{\bar{A}\bar{B}\bar{C}}\in \nsf{ppt}(\bar{A}:\bar{B}:\bar{C})} \ D(\rho_{\bar{A}\bar{B}\bar{C}} \| \sigma_{\bar{A}\bar{B}\bar{C}})
    \end{align}
    we find
    \begin{align}
        &I(A|\bar{A}:B|\bar{B}|C|\bar{C})_\rho+P(\bar{A}:\bar{B}:\bar{C})_{\rho}\nonumber\\
        &\quad \geq \sup_{X_{ABC},Y_{\bar{A}\bar{B}\bar{C}} > 0}\Big\{\tr\left[\rho_{A\bar{A}B\bar{B}C\bar{C}}\log X_{ABC}\otimes Y_{\bar{A}\bar{B}\bar{C}}\right]\nonumber\\
        &\quad\qquad-\log\tr\Big[\left(X_{ABC}\otimes Y_{\bar{A}\bar{B}\bar{C}}\right)\underbrace{\int_{-\infty}^{\infty} \textnormal{d}\beta_0(t)\left(\mathcal{R}_{\bar{A}\to A\bar{A}}^{[t]}\otimes\mathcal{R}_{\bar{B}\to B\bar{B}}^{[t]}\otimes\mathcal{R}_{\bar{C}\to C\bar{C}}^{[t]}\right)(\sigma_{\bar{A}\bar{B}\bar{C}})}_{=:\;\gamma_{A\bar{A}B\bar{B}C\bar{C}}} \Big]\Big\}\\
        &\quad =\sup_{X_{ABC},Y_{\bar{A}\bar{B}\bar{C}} > 0} \Big\{\tr\left[\rho_{ABC}\log X_{ABC}\right]+\tr\left[\rho_{\bar{A}\bar{B}\bar{C}}\log Y_{\bar{A}\bar{B}\bar{C}}\right]-\log\left(\tr\left[X_{ABC}\gamma_{ABC}\right]\cdot\tr\left[Y_{\bar{A}\bar{B}\bar{C}}\hat{\gamma}_{\bar{A}\bar{B}\bar{C}}\right]\right)\Big\}\\
        &\quad \geq P_{\nsf{PPT}}(A:B:C)_{\rho}+P_{\nsf{ALL}}(\bar{A}:\bar{B}:\bar{C})_{\rho}\,,
    \end{align}
    where we set
    \begin{align}
        \text{$\hat{\gamma}_{\bar{A}\bar{B}\bar{C}}:=\frac{\tr_{ABC}\left[X_{ABC}\gamma_{A\bar{A}B\bar{B}C\bar{C}}\right]}{\tr\left[X_{ABC}\gamma_{ABC}\right]}$ for the choice $X_{ABC}\in\nsf{PPT}(A:B:C)$,}
    \end{align}
    and used that $\gamma_{ABC}\in\nsf{ppt}(A:B:C)$ as well as $\hat{\gamma}_{\bar{A}\bar{B}\bar{C}}\in\nsf{ppt}(\bar{A}:\bar{B}:\bar{C})$ inherit the relevant PPT structure from the choice $X_{ABC}\in\nsf{PPT}(A:B:C)$, similarly as in the bipartite case. Again, the crucial point is that the recovery maps all commute. Eq.~\eqref{eq:cemi-multiIV} then follows from multipartite version of \cite[Theorem 1]{piani09} in the form
    \begin{align}
        P_{\nsf{PPT}}(A_1A_2:B_1B_2:C_1C_2)_{\rho}\geq P_{\nsf{PPT}}(A_1:B_1:C_1)_{\rho}+P_{\nsf{PPT}}(A_2:B_2:C_2)_{\rho}\,.
    \end{align}
\end{proof}

\begin{proof}[Proof of Proposition \ref{prop:squashed-recov-multi}]
We first prove Eq.~\eqref{eq:cqmi-3II} by writing
\begin{align}
&I(A_1:A_2:A_3|C)_{\rho}\nonumber\\
&\quad=D\left(\rho_{A_1A_2A_3C}\middle\|\exp(\log\rho_{A_1C}+\log\rho_{A_2C}+\log\rho_{A_3C}-2\log\rho_C)\right)\\
&\quad=\sup_{\omega_{A_1A_2A_3C} > 0}\Big\{\tr\left[\rho_{A_1A_2A_3C}\log\omega_{A_1A_2A_3C}\right]\nonumber\\
&\quad\qquad - \log \tr\left[\exp (\log\rho_{A_1C}+\log\rho_{A_2C}+\log\rho_{A_3C}-2\log\rho_C+ \log\omega_{A_1A_2A_3C}\right]\Big\}\\
&\quad\geq \sup_{\omega_{A_1A_2A_3C} > 0} \Bigg\{\tr\left[\rho_{A_1A_2A_3C}\log\omega_{A_1A_2A_3C}\right] \nonumber\\
&\quad\qquad - \log \tr\left[\omega_{A_1A_2A_3C}\int_{-\infty}^{\infty} \textnormal{d}\beta_0(t)\left(\mathcal{R}_{C\to A_3C}^{[t]}\circ\mathcal{R}_{C\to A_2C}^{[t]}\right)(\rho_{A_1C})\right]\Bigg\}\\
&\quad=D_{\nsf{ALL}}\left(\rho_{A_1A_2A_3C}\middle\|\int_{-\infty}^{\infty} \textnormal{d}\beta_0(t)\left(\mathcal{R}_{C\to A_3C}^{[t]}\circ\mathcal{R}_{C\to A_2C}^{[t]}\right)(\rho_{A_1C})\right)\,,\label{eq:ALL-firststepII}
\end{align}
where we employed the six matrix Golden-Thompson inequality for the Schatten two-norm from \cite[Corollary 3.3]{sutter16}. Next, Eq.~\eqref{eq:cqmi-3I} directly follows from Eq.~\eqref{eq:ALL-firststepII} via the additivity of CEMI on tensor product states, together with the asymptotic achievability of the measured relative entropy in Lemma \ref{lem:asymptotic-achiev}. Finally, for the proof of Eq.~\eqref{eq:cqmi-3III}, we follow the same consideration as in ~\cite[Appendix F]{sutter16}. Namely, for Peierls-Bogoliubov inequality as in Eq.~\eqref{eq:PB}, we choose $G_1=\rho_{A_1A_2A_3C}$ and
\begin{align}
G_2=\frac{1}{2}\left(-\log\rho_{A_1A_2A_3C}+\log\rho_{A_1C}+\log\rho_{A_2C}+\log\rho_{A_3C}-2\log\rho_C\right)\,,
\end{align}
which gives the chain of equations
\begin{align}
&I(A_1:A_2:A_3|C)_{\rho}\nonumber\\
&\geq-2\log\tr\left[\exp\left(\frac{1}{2}(-\log\rho_{A_1A_2A_3C}+\log\rho_{A_1C}+\log\rho_{A_2C}+\log\rho_{A_3C}-2\log\rho_C)\right)\right]\\
&\geq-2\int_{-\infty}^{\infty} \textnormal{d}\beta_0(t)\log\left\|\rho_{A_1A_2A_3C}^{\frac{1+it}{2}}\rho_{A_3C}^{\frac{1+it}{2}}\rho_C^{-\frac{1+it}{2}}\rho_{A_2C}^{\frac{1+it}{2}}\rho_C^{-\frac{1+it}{2}}\rho_{A_1C}^{\frac{1+it}{2}}\right\|_1\\
&=-\int_{-\infty}^{\infty} \textnormal{d}\beta_0(t)\log F\left(\rho_{A_1A_2A_3C},\left(\mathcal{R}_{C\to A_3C}^{[t]}\circ\mathcal{R}_{C\to A_2C}^{[t]}\right)(\rho_{A_1C})\right)\,,
\end{align}
where we subsequently employed the six matrix Golden-Thompson inequality for the Schatten one-norm from \cite[Corollary 3.3]{sutter16}.
\end{proof}


\bibliography{library_MB}
\bibliographystyle{ultimate}


\end{document}